\documentclass[english]{article}
\usepackage{url}
\usepackage{amssymb}
\usepackage{amsmath}
\usepackage{amsthm}
\usepackage{subfigure}
\usepackage{xspace}
\usepackage[normalem]{ulem}
 
\usepackage[T1]{fontenc}
\usepackage[latin9]{inputenc}

\makeatletter
\newtheorem{thm}{Theorem}
\newtheorem{cor}[thm]{Corollary}

  \newtheorem{defn}[thm]{Definition}
  \newtheorem{lem}[thm]{Lemma}
  \newtheorem{lem*}[thm]{Lemma}

\usepackage{graphicx,color}

\makeatother

\usepackage{babel}

\begin{document}
\global\long\def\Tasks{T}
\global\long\def\TasksP{T'}
\global\long\def\Def{:=}
\global\long\def\St{:\,}
\global\long\def\dright{\mathrm{RIGHT}}
\global\long\def\dleft{\mathrm{LEFT}}
\global\long\def\dn{\downarrow}
\global\long\def\up{\uparrow}

\providecommand{\eps}{\varepsilon}
\providecommand{\cT}{\mathcal{T}}
\providecommand{\cM}{\mathcal{M}}
\providecommand{\cA}{\mathcal{A}}
\providecommand{\cD}{\mathcal{D}}
\providecommand{\cR}{\mathcal{R}}
\providecommand{\cS}{\mathcal{S}}



\newcommand{\Crit}[1]{\ensuremath{C(#1)}}
\newcommand{\CritP}[1]{\ensuremath{C'(#1)}}
\newcommand{\CritLeft}[1]{\ensuremath{C_L(#1)}}
\newcommand{\CritRight}[1]{\ensuremath{C_R(#1)}}
\newcommand{\CritMed}[1]{\ensuremath{C_M(#1)}}
\newcommand{\CritCand}[1]{\ensuremath{\mathcal{C}(#1)}}
\newcommand{\CritCandLeft}[1]{\ensuremath{\mathcal{C}_L(#1)}}
\newcommand{\CritCandRight}[1]{\ensuremath{\mathcal{C}_R(#1)}}
\newcommand{\CritNum}{\ensuremath{\mathit{cr}}}
\newcommand{\CellVal}[1]{\mathit{OPT}(#1)}
\newcommand{\CUp}{C_{\up}}
\newcommand{\CDown}{C_{\dn}}
\newcommand{\CellAreaTasks}[1]{T(#1)}
\newcommand{\Line}[2]{\ensuremath{#1\!\!-\!\!#2}\xspace}

\newcommand{\abv}{abv}
\newcommand{\crit}{crit}
\newcommand{\subc}{subc}
\newcommand{\bound}{bound}

\newcommand{\Bf}{\boldsymbol{f}}
\newcommand{\Be}{\boldsymbol{e}}

\newcommand{\MUp}{m_{\up}}
\newcommand{\MDown}{m_{\dn}}

\newcommand{\Card}[1]{\left\lvert{#1}\right\rvert}


\graphicspath{{./figures/}}

\title{A Mazing 2+$\ensuremath{\mathbf{\epsilon}}$ Approximation \\ for Unsplittable
Flow on a Path}

\author{Aris Anagnostopoulos\footnote{Sapienza University of Rome, \texttt{aris@dis.uniroma1.it}} , 
Fabrizio Grandoni\footnote{IDSIA, University of Lugano, \texttt{fabrizio@idsia.ch}} , \\
Stefano Leonardi\footnote{Sapienza University of Rome, \texttt{leon@dis.uniroma1.it}} , and
Andreas Wiese\footnote{Max-Planck-Institut f\"ur Informatik, Saarbr\"ucken, \texttt{awiese@mpi-inf.mpg.de}. This work was partially supported by a fellowship within the Postdoc-Programme of the German Academic Exchange Service (DAAD).}}

\date{}

\maketitle

\begin{abstract}
We study the unsplittable flow on a path problem (UFP) where we are given
a path with non-negative edge capacities and tasks, which are
characterized by a subpath, a demand, and a profit. The goal is to find
the
most profitable subset of tasks whose total demand does not violate the
edge capacities. This problem naturally arises in many settings such
as bandwidth allocation, resource constrained scheduling, and interval
packing. 

A natural task classification defines the size of a task $i$ to be the ratio $\delta$ between the demand of $i$ and the minimum capacity of any edge used by $i$. If all tasks have sufficiently
small $\delta$, the problem is already well understood and there is a $1+\eps$ approximation. For the complementary
setting---instances whose tasks all have large $\delta$---much remains
unknown, and the best known polynomial-time procedure gives only (for any constant $\delta>0$) an approximation
ratio of $6+\eps$.

In this paper we present a polynomial time $1+\eps$ approximation for the latter setting.
Key to this result is a complex geometrically inspired 
dynamic program. Here each task is represented as a segment underneath
the capacity curve, and we identify a proper maze-like structure so that
each \emph{passage} of the maze is \emph{crossed} by only $O(1)$ tasks
in the computed solution.
In combination with the known PTAS for $\delta$-small tasks, our result implies a $2+\eps$ approximation for UFP, improving on the previous best $7+\eps$ approximation~[Bonsma et al., FOCS 2011].
We remark that our improved approximation factor matches the best known approximation ratio for the 
considerably  easier special case of uniform edge capacities.
 \end{abstract}


\section{Introduction}
\label{sec:intro}

In the \emph{unsplittable flow on a path} problem (UFP) we are given
a set of tasks $\Tasks$ and a path $G=(V,E)$.
Each edge~$e$ has a capacity $u_{e} \in \mathbb{N}^+$. Each task $i\in\Tasks$
is specified by a subpath $P(i)$ between the start (i.e. leftmost) vertex $s(i)\in V$ and the end 
(i.e. rightmost) vertex~$t(i)\in V$, a demand $d(i)> 0$, and a profit (or weight) $w(i)\geq 0$.  For each edge
$e\in E$, denote by $\Tasks_{e}$ all tasks $i$ using $e$, i.e. such that $e\in P(i)$. 
For a subset of tasks $T'$, let $w(\Tasks'):=\sum_{i\in\Tasks'}w(i)$ and $d(\Tasks'):=\sum_{i\in\Tasks'}d(i)$.
The goal is to select
a subset of tasks $\Tasks'$ with maximum profit $w(\Tasks')$ such that
$d(\Tasks'\cap\Tasks_{e})\le u_{e}$,
for each edge $e$.

The problem and variations of it are motivated by several applications
in settings such as bandwidth allocation, interval packing, multicommodity
flow, and scheduling. For example, edge capacities might model a given resource whose supply varies over a time horizon. Here,  
tasks correspond to jobs with
given start- and end-times and each job has a fixed demand for the mentioned resource. The goal is then to select the most
profitable subset of jobs whose total demand at any time can be satisfied with the available resources. 

When studying the problem algorithmically, a natural classification
of the tasks is the following: W.l.o.g., let us assume that edge
capacities are all distinct (this can be achieved by slight
perturbations and scaling, see~\cite{BSW11}). For each task $i$
define its \emph{bottleneck
capacity} $b(i):=\min\{u_{e}: e\in P(i)\}$. Let also the \emph{bottleneck edge} $e(i)$ of $i$ be the edge of $i$ with capacity $b(i)$. W.l.o.g. we can assume $d(i)\leq b(i)$, otherwise task $i$ can be discarded. For any value $\delta\in(0,1]$
we say that a task $i$ is $\delta$-large if $d(i)\geq \delta\cdot b(i)$
and \emph{$\delta$-small} otherwise. 

If all tasks are $\delta$-small ($\delta$-small instances), then the
problem is well understood. In particular, by applying the LP-rounding
and grouping techniques from~\cite{BSW11}, we immediately obtain the
following result (see the appendix for its proof).
\begin{lem}\label{lem:ptasSmall}
For any $\eps>0$ there is a $\delta\in (0,1]$ such that in polynomial time one can compute a $(1+\eps)$-approximation for $\delta$-small instances of UFP.
\end{lem}
 
However, much remains unclear for the complementary case where all tasks are $\delta$-large ($\delta$-large instances),  even if $\delta$ is very close to $1$.
Importantly, for such instances the canonical LP has an integrality
gap of $\Omega(n)$~\cite{CCGK2007}. The best known approximation factor 
for this setting is $2k$, where $k\in \mathbb{N}$ such that 
$\delta>\frac{1}{k}$ (and in particular $k\ge 2$)~\cite{BSW11}.
Bonsma et al.~\cite{BSW11} reduce the problem to an instance of maximum independent set of rectangles and this approach inherently loses a factor of $2k\geq 4$ in the approximation ratio.
The best known
$(6+\eps)$-approximation algorithm for $\delta$-large instances, for
any $\delta>0$, combines the approach above (with $k=2$) with another
algorithm, which is $2+\eps$ approximate for instances that are
$1/2$-small and $\delta$-large at the same time.
Combining this $(6+\eps)$ approximation with the result from Lemma \ref{lem:ptasSmall}, one obtains the currently best $(7+\eps)$-approximation algorithm for UFP on general instances \cite{BSW11}.

\subsection{Related Work}

As said above, the best known polynomial time approximation algorithm
for unsplittable flow on a path achieves an approximation factor of $7+\eps$~\cite{BSW11}.
This result improves on the previously best known polynomial time
$O(\log n)$-approximation algorithm designed by Bansal et
al.~\cite{SODA-unsplit-flow}.
When allowing more running time, there is a quasi-PTAS, that is, a $(1+\eps)$-approximation running in 
$O(2^ {\mathrm{polylog}(n)})$ time that additionally assumes a
quasi-polynomial bound on the edge capacities and the
demands of the input instance~\cite{BCES2006}. 
In terms of lower bounds, the problem is strongly NP-hard, even in the 
case of uniform edge capacities and unit profits \cite{BSW11,CWMX-ESA2010,DPS2010}.

The canonical LP-relaxation suffers from a $\Omega(n)$ integrality gap~\cite{CCGK2007}.
Adding further constraints, Chekuri, Ene, and Korula give an LP relaxation with an integrality gap
of only $O(\log^2 n)$~\cite{CEKApprox2009}, which was recently improved to $O(\log n)$~\cite{CEKunp}.

Because of the difficulty of the general problem, researchers have
studied special cases. A very common
assumption is the \emph{no-bottleneck assumption} (NBA), which requires that $\max_i \{d(i)\} \le \min_e \{u_e\}$. 
Chekuri, Mydlarz, and Shepherd~\cite{CMS07} give the currently best known
$(2+\eps)$-approximation algorithm under NBA.
Note that this matches the best known result for the further restricted
case of uniform edge capacities of Calinescu et al.~\cite{CCKR2002}.

When generalizing the problem to trees, Chekuri et al.~\cite{CEKApprox2009} give a $O( \log(1/\gamma) /
\gamma^3)$-approximation
algorithm for the case that all tasks are $(1-\gamma)$-small. Under the NBA, Chekuri et al.~\cite{CMS07} design a
48-approximation algorithm.  Note that on trees the problem
becomes APX-hard, even for unit demands, edge capacities being either 1
or 2, and trees with depth three~\cite{GVY1997}.

On arbitrary graphs, UFP generalizes the well-known Edge Disjoint Path
Problem (EDPP). On directed graphs, there is a
$O(\sqrt{\Card{E}})$ approximation algorithm by Kleinberg~\cite{Kleinberg96},
which matches the lower bound of $\Omega(\Card{E}^{1/2-\eps})$
by Guruswami et al.~\cite{GuruswamiKhanna2003}. When assuming the NBA,
Azar, and Regev~\cite{AzarRegev2006} give an
$O(\sqrt{\Card{E}})$ approximation algorithm for UFP. On the other hand, they
show that without the NBA the problem is
NP-hard to approximate within a factor better than $O(\Card{E}^{1-\eps})$.

\subsection{Our Contribution}

In this paper, we present a polynomial-time approximation scheme for
$\delta$-large instances of UFP (for any constant $\delta>0$), improving
on the previous best $6+\eps$ approximation~\cite{BSW11}.
We remark that instances with only $\delta$-large tasks might be relevant in practice.

Furthermore, in combination with the algorithm from
Lemma~\ref{lem:ptasSmall}, our PTAS implies a $2+\eps$ approximation
for arbitrary UFP instances,
without any further assumptions (such as the NBA or restrictions on edge capacities).
This improves on the previous best $7+\eps$ approximation for the
problem~\cite{BSW11}. Note that our $2+\eps$ approximation matches
the best known approximation factor for the considerably easier special
case of uniform edge-capacities \cite{CCKR2002}, where, in particular,
the canonical LP has an integrality gap of a small constant and the
$\delta$-large tasks can be handled easily with a straightforward
dynamic program (DP). Therefore, we close the gap in terms of (known)
polynomial-time approximation ratios between the uniform and general
case.

\begin{figure}
\centerline{\scalebox{.25}{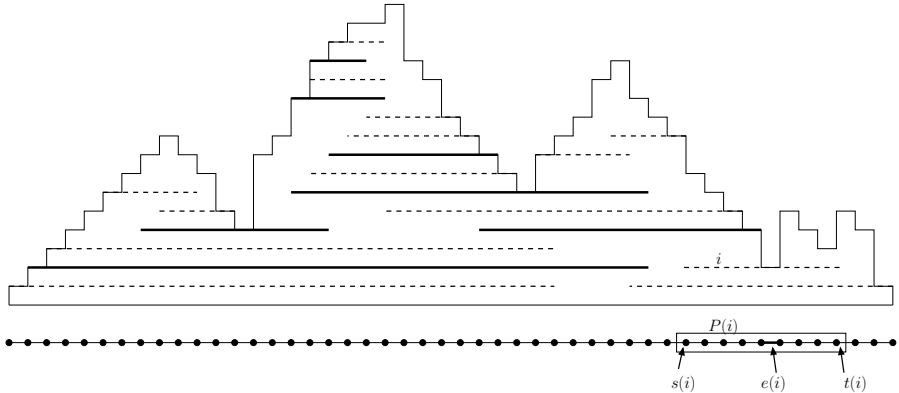}}
\caption{
An example of capacity curve, with segments associated to some tasks $T'$ (dashed) and m-tasks $M'$ (bold). Note that $(T',M')$ is 2-thin.
}
\label{fig:maze}
\end{figure}

When solving general UFP, large tasks are difficult to handle: as
mentioned above, the canonical LP-relaxation suffers from an integrality
gap of $\Omega(n)$~\cite{CCGK2007}. Also, in contrast to
the NBA-case there is no canonical polynomial time dynamic program
for them since the number of large tasks per edge can be up to $n$
(whereas under the NBA it is $O(1/\delta)$, see~\cite{CMS07}).

In~\cite{BSW11} a dynamic program for large tasks was presented; however,
their reduction to maximum independent set in rectangle intersection
graphs inherently loses a factor of 4 in the approximation ratio,
already for $\frac{1}{2}$-large tasks (and even a factor of $2k$
for $\frac{1}{k}$-large tasks). 

Our PTAS for $\delta$-large instances is also a dynamic program,
but it deviates substantially from the DP-approaches above. We exploit
the following geometric viewpoint. We represent capacities with a
closed curve on the 2D plane (the \emph{capacity curve}) as follows:
Let us label nodes from $1$ to $n$ going from left to right. For
each edge $e=(v,v+1)$, we draw a horizontal line segment (or segment
for short) $[v,v+1]\times\{u_{e}\}$. Then we add a horizontal segment
$[1,n]\times\{0\}$, and vertical segments in a natural way to obtain
a closed curve. We represent each task $i$ as a horizontal segment
$(s(i),t(i))\times\{b(i)\}$. In particular, this segment is contained
in the capacity curve, and touches the horizontal segment corresponding
to its bottleneck edge (see Figure~\ref{fig:maze}).

Canonically, one might want to traverse the path from left to right
and introduce DP-cells encoding all possible choices for certain subpaths.
Instead, we traverse the area underneath the capacity profile using
the above geometric representation, going from the root in the bottom
left to the dead-ends of the maze (the leaves). To guide this traversal,
we use some tasks which we call \emph{maze tasks} below. Those tasks
fulfill two functions: they structure the area within the capacity
curve into a maze with a tree topology such that, intuitively, each
\emph{passage} of the maze is \emph{crossed} by only $k=O(1)$ tasks (see Figure~\ref{fig:maze}).
We will refer to this property as $k$\emph{-thinness} later. It will
be crucial for bounding the number of DP-cells. 

One still arising difficulty is that when traversing the maze in higher
regions, we cannot afford to remember precisely which tasks were selected
in lower regions. To this end, the maze tasks have a second function.
We use them to make it affordable to ``forget'' some decisions while
moving from the root to the leaves. Before returning a solution we
remove all maze tasks from the computed set. By constructing our algorithm carefully,
we ensure that the capacity of each maze task $m$ compensates for
the information we allowed to ``forget'' due to $m$. We call a
solution \emph{weakly feasible} if it balances the latter correctly\emph{.
}So our DP computes a weakly feasible $k$-thin pair $(T',M')$ where
$T'\subseteq T$ is a set of tasks and $M'$ is a set of maze tasks.
The final output consists only of $T'$ whose weight we seek to maximize,
and the DP computes the optimal solution among all pairs $(T',M')$.

Since at the end we will remove the maze tasks of a computed solution
we need to ensure that there is in fact a solution $(T',M')$ where
the weight of $T'$ is close to the optimum~$T^{*}$. This is proved
by a non-trivial sequence of reductions where eventually the tasks
of $T^{*}$ are mapped into directed paths of a proper rooted tree.
On those paths we define a min-flow LP where each integral solution
induces a $k$-thin pair $(T',M')$ where the tasks $T^{*}$ are partitioned
into $T'$ and $M'$. The objective is to minimize $w(M')$. The claim
then follows by showing that there exists a cheap fractional solution of weight at most $\eps\cdot w(T^*)$,
and that the LP matrix is totally unimodular.

\section{Overview of the Algorithm}
\label{ref:overview}

In this section we describe our methodology, which results in a
polynomial-time $(1+\eps)$-approximation algorithm for
$\delta$-large UFP instances  (for any two given constants
$\eps,\delta>0$). After running a polynomial time preprocessing routine we
can assume that each vertex is either the start or the end vertex of exactly
one task in $T$ (similarly as in~\cite{BCES2006}). 
Thus, the number $n$ of nodes in the graph is $\Theta(T)$.

Now we define the \emph{maze tasks}, or m-tasks for short,
which we use to structure our solution.
For each pair of tasks $i$ and~$j$ that share the same
bottleneck edge $e$ (possibly $i=j$), we define an m-task $m$ with
$P(m)=P(i)\cup P(j)$. Analogously to regular tasks, we set $b(m)=u_e$ and $e(m)=e$. 
Furthermore, we define $d(m)=\delta\cdot u_e$ and $w(m)=0$. Let $M_e$ be the m-tasks $m$ with $e\in P(m)$. 
Note that, by the above preprocessing, no two m-tasks with different bottleneck capacity share the same endpoint.

Our goal is to search for solutions in the form of \emph{maze pairs}
$(T',M') \in 2^T\times 2^M$, where we require for any two different
tasks $m,m' \in M'$ that $b(m')\neq b(m'')$. Let  $k=k(\eps,\delta)$ be
a proper integer constant to be defined later. We restrict our attention
to maze pairs that are \emph{$k$-thin} and \emph{weakly feasible}, as
defined below.

Intuitively, a maze pair $(T',M')$ is $k$-thin if, for any edge~$e$, between two
consecutive line segments associated to m-tasks from $M'\cap M_e$ there are at most $k$
segments associated to tasks in $T'\cap T_e$ (see Figure~\ref{fig:maze}).
\begin{defn} [$k$-thinness]
A maze pair $(T',M')$ is $k$-thin if for every edge $e$ and every set $T'' \subseteq T' \cap T_e$ with
$\Card{T''} > k$ there is an m-task $m \in M' \cap M_e$ such that
$\min_{i\in T''}\{b(i)\}  \le b(m) < \max_{i\in T''}\{b(i)\}$.
\end{defn}

In Section~\ref{sec:construct-maze} we prove that, for large enough $k$,
there exists a $k$-thin maze pair $(\tilde{T},\tilde{M})$ so that
$\tilde{T}$ is a good approximation to the optimum $T^*$ and
$\tilde{T}\cup \tilde{M}$ is feasible (i.e., $d(\tilde{T}\cap T_e)+d(\tilde{M}\cap M_e)\leq u_e$ on each edge $e$).
\begin{lem}\label{lem:apx}
For any $\eps,\delta>0$ there is a $k\in \mathbb{N}$, such that
for any $\delta$-large instance of UFP, there exists a $k$-thin
maze pair $(\tilde{T},\tilde{M})$ such that $w(\tilde{T})\geq
(1+\eps)^{-1}w(T^*)$ and $\tilde{T}\cup \tilde{M}$ is feasible.
\end{lem}

However, we are not able to compute the most profitable $k$-thin maze
pair in polynomial time. For this reason we relax the notion of
feasibility of a maze pair $(T',M')$ so that $T'\cup M'$ might not be
feasible, but still $T'$ alone is feasible (which is sufficient for our
purposes). We need some definitions first. For every m-task $m\in M$
and any subset of tasks $T'$, we partition the set 
$T'(m):=\{i\in T' : P(i)\cap P(m)\ne\emptyset\}$ of tasks of $T'$ sharing some edge with $m$ into three (disjoint) subsets:

\begin{itemize}\itemsep0pt
\item[$\bullet$] {\bf (above tasks)} $\abv(m,T'):=\{i\in T'(m) : b(i)> b(m)\}$.
\item[$\bullet$] {\bf (critical tasks)} $\crit(m,T'):=\{i\in T'(m) : b(m)\geq b(i)\geq \frac{\delta}{2}b(m)\}$.
\item[$\bullet$] {\bf (subcritical tasks)} $\subc(m,T'):=\{i\in T'(m) : b(i)< \frac{\delta}{2}b(m)\}$.
\end{itemize}
We also define $\abv_e(m,T'):=\abv(m,T')\cap T_e$, and we define
analogously $\crit_e(m,T')$ and $\subc_e(m,T')$.
\begin{defn}[Weak feasibility]\label{dfn:weak-feasibility}
We define a maze pair $(T',M')$ to be \emph{weakly feasible} if for every edge $e$ it holds that $d(\abv_e(m_e,T'))+d(\crit_e(m_e,T'))+d(m_e)\leq u_e$, where
$m_e$ is the m-task in $M'\cap M_e$ of largest bottleneck capacity, or 
$d(T'\cap T_e)\leq u_e$, if $M'\cap M_e=\emptyset$.
\end{defn}
Next we show that weak feasibility of a maze pair $(T',M')$ implies feasibility of $T'$.
\begin{lem}\label{lem:weak-feas-global-feas}
Let $(T',M')$ be a weakly feasible maze pair. Then $T'$ is feasible.
\end{lem}
\begin{proof}
We order the edges by their capacities in non-decreasing order. Assume
w.l.o.g.~that this order is given by $\{e_{1},...,e_{m}\}$. 
We prove the claim by induction on the index~$j$. More precisely,
we prove that $d(T' \cap T_{e_j}) \le u_{e_j}$ for all $j$.

Consider first $e_1$. If there is no m-task using $e_1$, then the claim
is true by definition. Otherwise let $m_1=m_{e_1}$ be the (only) m-task in $M'\cap M_{e_1}$. All tasks $i\in T'\cap T_{e_1}$ must have $b(i)=u_{e_1}$ (in particular, they are critical for $m_1$). Thus 
$d(T'\cap T_e)\leq d(T'\cap T_e)+d(m_{1})=d(\abv_{e_1}(m_{1},T'))+d(\crit_{e_1}(m_{1},T'))+d(m_{1})\leq u_e$.

Now suppose by induction that there is a value $j\in\mathbb{N}$ such
that $d(T'\cap T_{e_{j'}})\le u_{e_{j'}}$ for all
$j'\in\{1,...,j-1\}$. Consider the edge $e_{j}$. Once again if there is no m-task using $e_j$, then the claim is true by definition.
Otherwise, let $m_j:=m_{e_j}$.
Consider the subcritical tasks $SC:=\subc_{e_j}(m_j,T')$. By definition, $e_j$ is not the bottleneck edge of any task in $SC$.
We partition $SC$ into the sets $SC_{L}$ and $SC_{R}$, containing
the tasks with bottleneck edge on the left of $e_j$ and on the right
of $e_j$, respectively. Consider the set $SC_{L}$. Let $i_L\in SC_{L}$
be a task with maximum bottleneck capacity in $SC_{L}$ and let $e_L$
all tasks
in $SC_{L}$ use $e_{L}$ and $u_{e_{L}}=b(i_{L})<\frac{\delta}{2}\cdot b(m_j)$.
Using the induction hypothesis on $e_L$, we obtain that $d(SC_{L})=d(SC_{L}\cap T_{e_{L}})\le d(T'\cap T_{e_{L}})\le u_{e_{L}}<\frac{\delta}{2}\cdot b(m_j)$.
Similarly, we obtain that $d(SC_{R})<\frac{\delta}{2}\cdot b(m_j)$.
Since $d(m_j)=\delta\cdot b(m_j)$ the m-task $m$ compensates for all
tasks in $SC$, that is, $d(SC)=d(\subc_{e_j}(m_j,T'))\leq d(m_j)$. Hence 
\begin{eqnarray*}
d(T'\cap T_{e_{j}}) & \hspace{-10pt}= & \hspace{-8pt} d(\abv_{e_j}(m_j,T'))+d(\crit_{e_j}(m_j,T'))\\
& &\hspace{-10pt} + \; d(\subc_{e_j}(m_j,T')) \\
& \hspace{-10pt}\leq &\hspace{-8pt} d(\abv_{e_j}(m_j,T'))+d(\crit_{e_j}(m_j,T'))+d(m_j)\\
& \hspace{-10pt}\leq &\hspace{-8pt} u_{e_j},
\end{eqnarray*}
where the last inequality follows from the weak feasibility of $(T',M')$.
\end{proof}

Note that, by definition, the maze pair $(\tilde{T},\tilde{M})$ obtained in
Lemma~\ref{lem:apx} is also weakly feasible.
In Section~\ref{sec:DP} we present a polynomial-time dynamic program
that computes the weakly feasible $k$-thin 
maze pair with highest profit.

\begin{lem}\label{lem:DP}
For any constants $\delta,k>0$, $k\in \mathbb{N}$, there is a polynomial-time dynamic program that computes a weakly
feasible $k$-thin maze pair $(T',M')$ of largest profit $w(T')$.
\end{lem}

A crucial property that we exploit in the design of our dynamic program
is that for each m-task in a weakly feasible maze pair
the number of critical tasks is bounded by a constant depending only on $\delta$.
\begin{lem}\label{lem:critical}
Let $(T',M')$ be a weakly feasible maze pair and $m\in M'$. It holds that
$\Card{\crit(m,T')}\leq ncrit(\delta):=\frac{4}{\delta^2}+\frac{1}{\delta}$.

\end{lem}
\begin{proof}
First recall that by Lemma~\ref{lem:weak-feas-global-feas}, $T'$ is a
feasible solution. Consider the tasks $i\in\crit(m,T')$ with $b(i)=b(m)$. Because all tasks are $\delta$-large, there can be
at most $1/\delta$ such tasks. The remaining tasks $i\in \crit(m,T')$ have $b(i)<b(m)$ and must use the leftmost edge $e_L$ of
$P(m)$ or the rightmost edge $e_R$ of $P(m)$ (or both). Consider the
tasks $C_L$ of the first type: we will show that $|C_L|\leq 2/\delta^2$.
A symmetric argument holds for the remaining tasks $C_R$, hence giving
the claim. Consider the task $i_L\in C_L$ that has the largest 
$b(i_L)$. By the definition of $C_L$ and $i_L$ all tasks in $C_L$ must use $e(i_L)$ and $b(i_L)\le b(m)$. Each task $i\in C_L$ is critical
for $m$ and thus $b(i)\geq \frac{\delta}{2}b(m)$. Also, $i$ is
$\delta$-large and so $d(i)\geq \delta b(i)\geq \frac{\delta^2}{2}b(m)$. Therefore, there can be at most
$b(i_L)/(\frac{\delta^2}{2}b(m))\leq \frac{2}{\delta^2}$ such tasks.
\end{proof}

By combining Lemmas~\ref{lem:apx}, \ref{lem:weak-feas-global-feas}
and~\ref{lem:DP} we obtain the main theorem of this paper.

\begin{thm}\label{thr:ptas}
For any constant $\delta>0$, there is a PTAS for $\delta$-large
instances of UFP.
\end{thm}
Combining Theorem \ref{thr:ptas} with Lemma~\ref{lem:ptasSmall}, we obtain the following corollary.
\begin{cor}
There is a polynomial time $(2+\eps)$ approximation algorithm
for UFP. 
\end{cor}

\section{A Thin Profitable Maze Pair}\label{sec:construct-maze}

In this section we prove Lemma~\ref{lem:apx}:
we present a procedure that, given an (optimal) solution $T^*$,
carefully selects some of the
tasks from $T^*$ and replaces them with m-tasks from $M$ that use the
same edges. The tasks are selected in such a way that the tasks removed
from $T^*$ have small weight, and at the same time the 
resulting maze pair is $k$-thin.
Hence, our proof is constructive and even leads to a polynomial time
algorithm; note however that for our purposes a non-constructive
argumentation would be sufficient.

\begin{figure*}
\centering
\subfigure[\label{fig:maze-construction-initial}Initial instance. The nubers identify some of the
  tasks.]{\scalebox{.13}{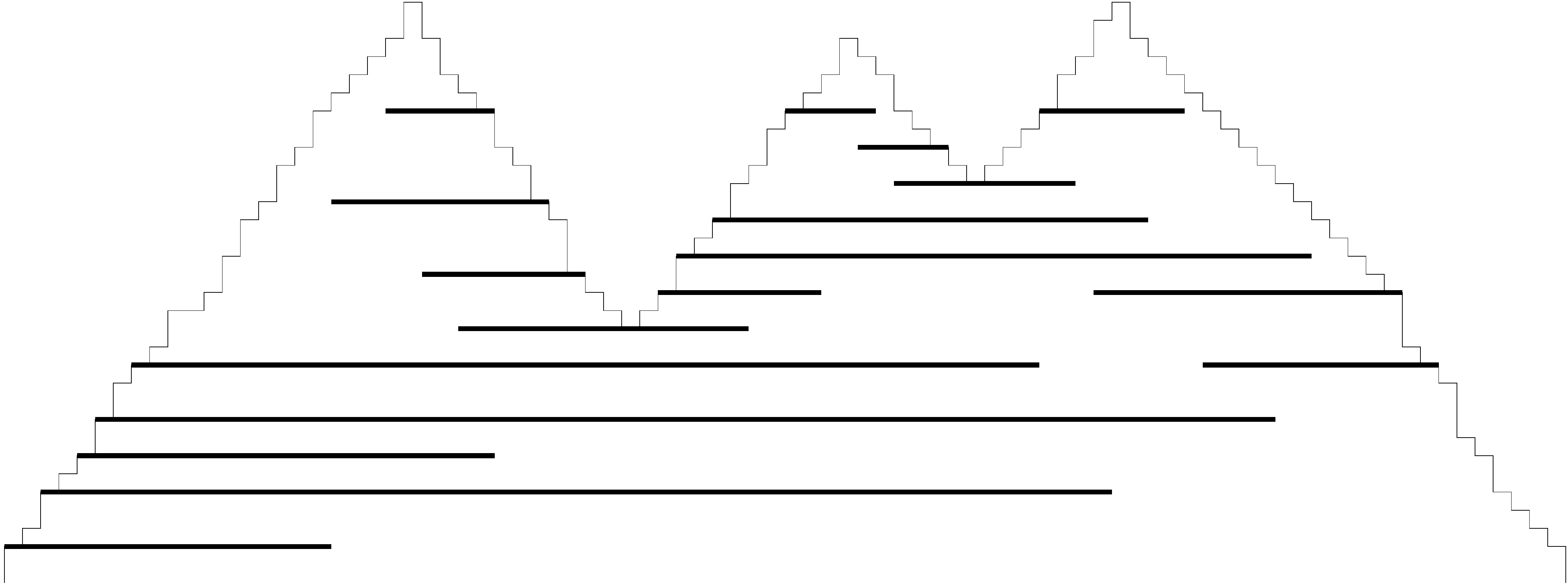}}\hspace{.5cm}
\subfigure[\label{fig:maze-construction-split}Splitting the line segments to the ones in $L_{R}$ (in bold)
  and to $L_{L}$
  (dashed).]{\scalebox{.13}{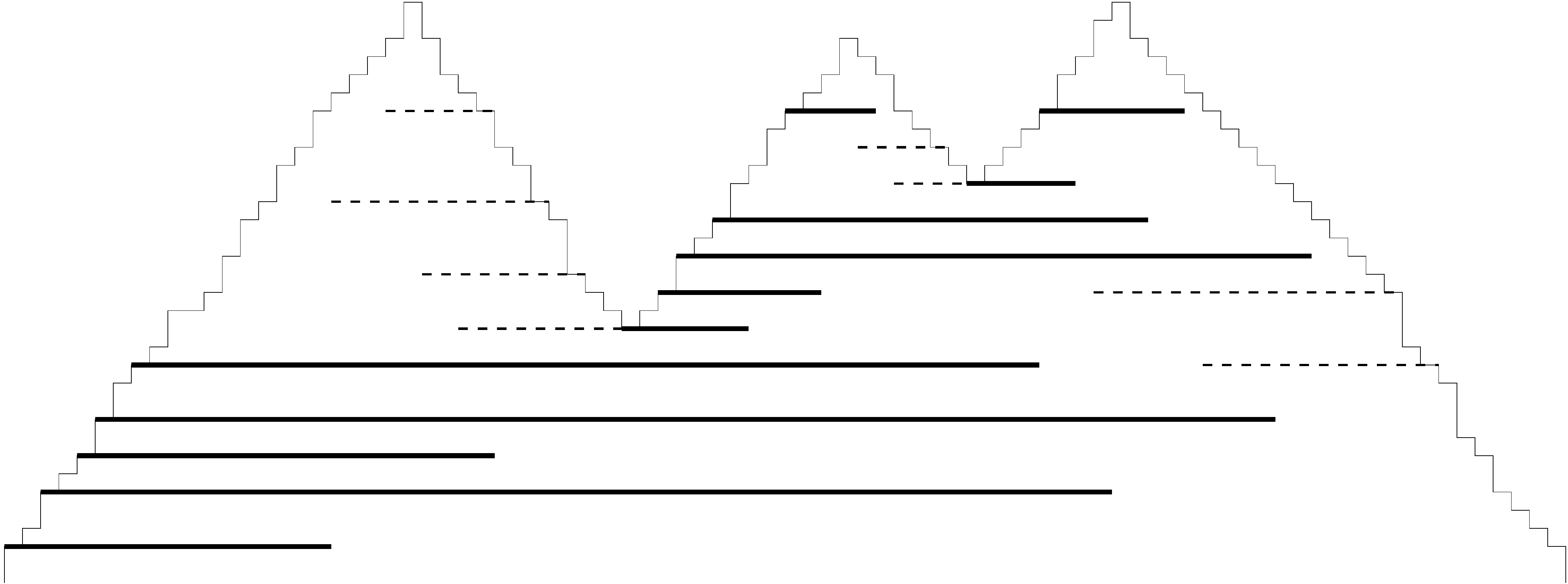}}\hspace{.5cm}

\mbox{

\subfigure[\label{fig:maze-construction-shift}Line segments in $L_{R}$ shifted up (distances distorted) and extended to the
left (in dashed).]{\scalebox{.15}{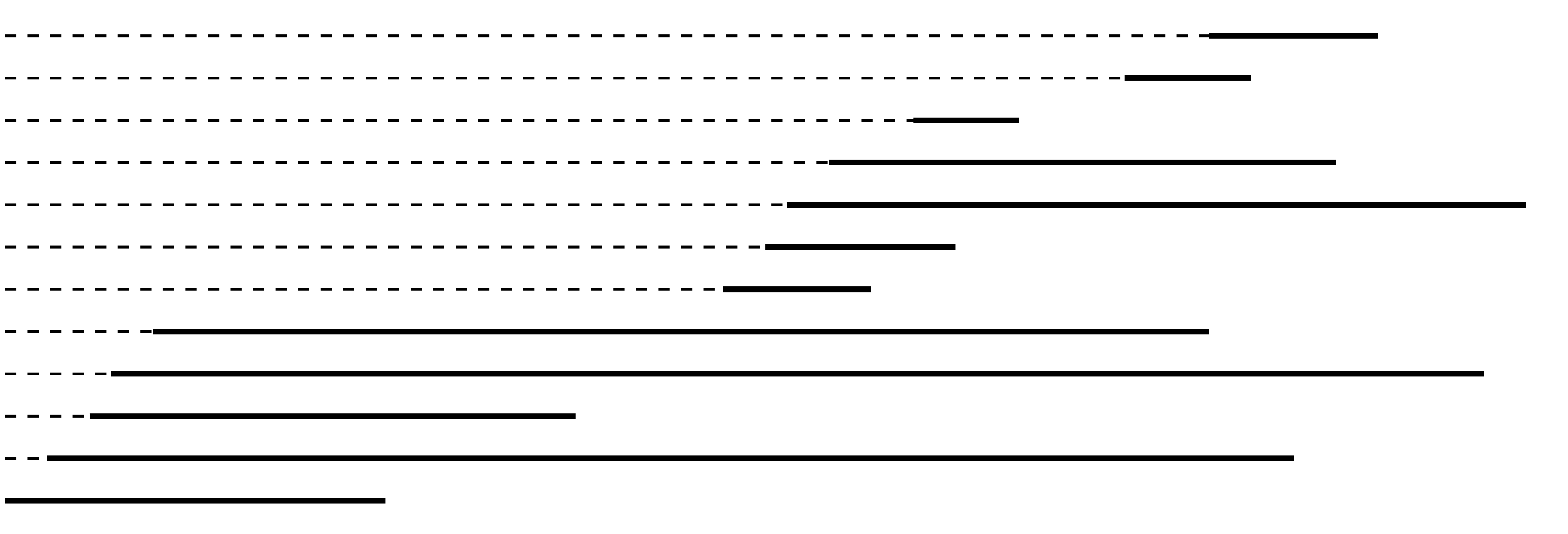}}

\subfigure[\label{fig:maze-construction-decomposition}
Decomposition of the segments for $k=4$. The dashed lines indicates the values ($y_0$ in
the text) where the respective interval $I_w$ is split. 
]
 {\hspace{.7cm}\scalebox{.22}{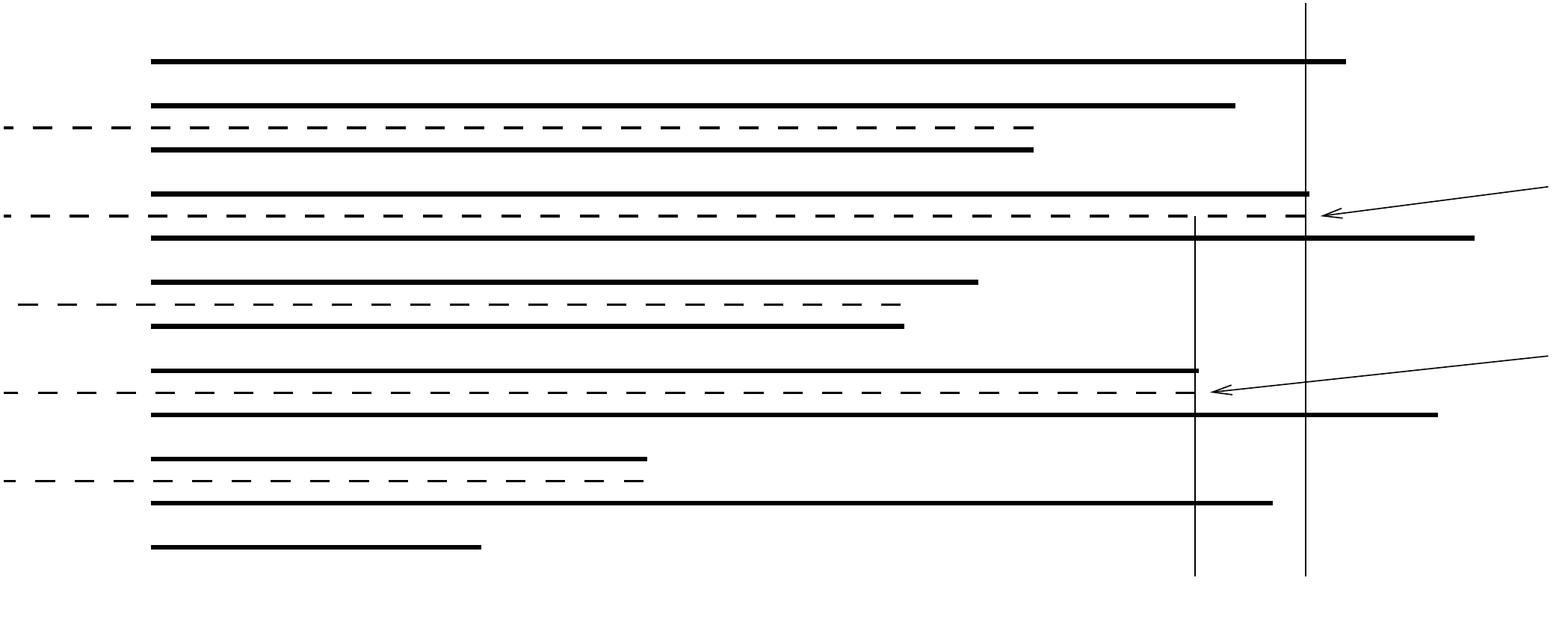}\hspace{.7cm}}\hspace{1cm}
}
\subfigure[\label{fig:maze-construction-tree}The tree created from the
decomposition of (d) (for representation issues, arc directions are omitted and some nodes of degree $2$ are contracted). 
A and B indicate the nodes corresponding to the splittings in (d). 
]{\hspace{.5cm}\scalebox{.19}{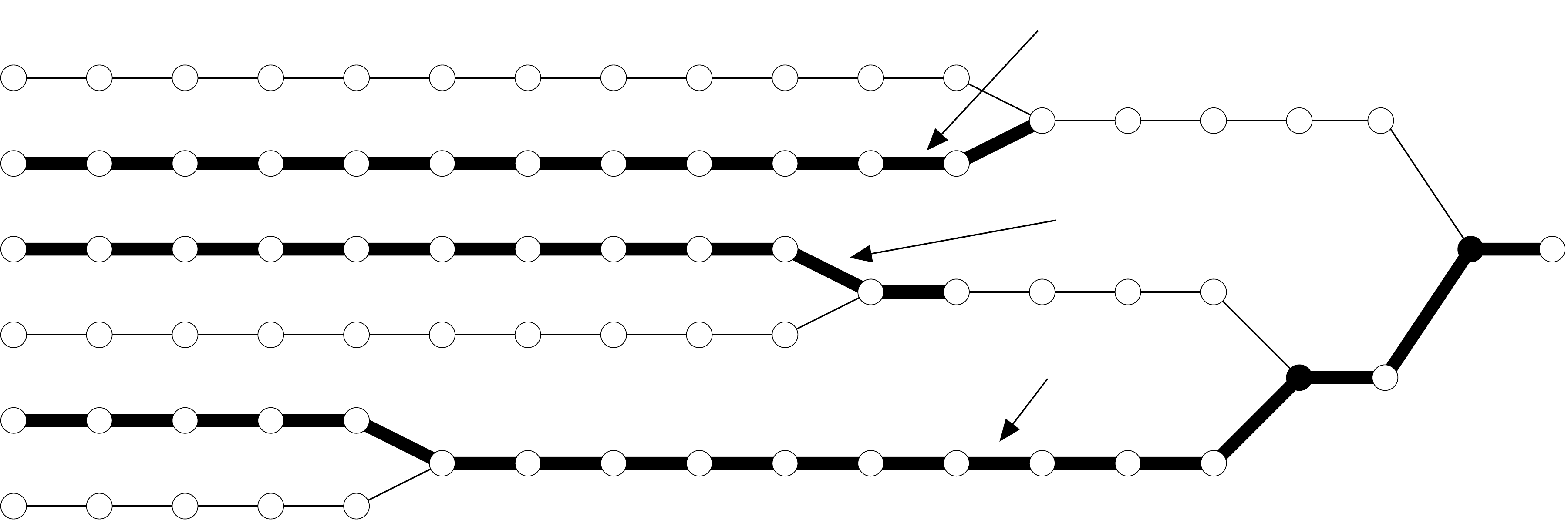}\hspace{.5cm}}
\caption{Construction of the maze}
\label{fig:maze-construction}
\end{figure*}

Let $T^*$ be the optimal solution for the instance under consideration.
Let $\ell(i)$ denote the line segment $(s(i),t(i))\times \{b(i)\}$ associated to each task $i \in T^*$. Define $L:=\{\ell(i) : i\in T^*\}$ and $w(\ell(i)):=w(i)$. We say that a segment $(a,b)\times\{y\}$ \emph{contains an edge $e=(v,v+1)$} if $(v,v+1)\subseteq (a,b)$ where
we assume that the $n$ vertices of the graph are labeled by $1,...,n$ from left to right.

We want to select a subset $L'\subseteq L$ such that 
$w(L'):=\sum_{\ell(i)\in L'}w(\ell(i))$ is at most $\eps\cdot w(T^{*})$
and any vertical segment $\{x\}\times (y_b,y_t)$ intersecting more than
$k$ segments in $L$ intersects at least one segment in $L'$.
We call a set $L'$ with the latter property \emph{$k$-thin} \emph{for
$L$}. 

As we will show, for proving Lemma~\ref{lem:apx} it suffices to
find a $k$-thin set $L'$ for $L$ because of the following transformation of
$L'$ into a maze-pair $(T(L'),M(L'))$. We define
$T(L'):=\{i: \ell(i)\in L\setminus L'\}$.
For constructing $M(L')$ we group the lines in $L'$ according to
the bottleneck edges of their corresponding tasks. For each edge $e$,
we define $L'_{e}:=\{\ell(i)\in L': e(i)=e\}$. Now for each edge $e$
with $L'_{e}\ne\emptyset$ we add an m-task $m_{e}\in M$ into $M(L')$
whose endpoints are the leftmost and rightmost node of the 

path $\cup_{\ell(i)\in L'_e}P(i)=P(i_L)\cup P(i_R)$ for the task $i_L \in L'_e$
with leftmost start vertex and the task $i_R \in L'_e$ with rightmost end vertex
(in a sense, we glue $i_L$ and $i_R$ together to form an m-task). Observe that, as
required in the definition of a maze pair,
$b(m')\neq b(m'')$ for any two distinct $m',m''\in M(L')$.

\begin{lem}
\label{lem:segments} If a set $L'\subseteq L$ is $k$-thin for $L$, then the maze pair
$(T(L'),M(L'))$ is $(k+1/\delta)$-thin and $T(L')\cup M(L')$ is 
feasible. 
\end{lem}
\begin{proof}
Consider any edge $e=(u,u+1)$, and any set of $k+1/\delta+1$ tasks
$T'\subseteq T^*\cap T_e$. Let $\{i_{1},i_{2},\ldots,i_{k+1}\}\subseteq T'$ be
$k+1$ of them with lowest bottleneck capacity, in non-decreasing
order of bottleneck capacity. Since $T^{*}$ is feasible, and since the
tasks in $T^*$ are $\delta$-large, there
cannot be more than $1/\delta$ tasks in $T'$ of bottleneck capacity equal
to $b_{max}:=\max_{i\in T'}\{b(i)\}$. It follows that $b(i_{j})<b_{max}$
for all $1\leq j\leq k+1$. Consider a vertical line segment
$\ell'$ with $x$-coordinate $u+\frac{1}{2}$
that intersects $\ell(i_{1}),\dots,\ell(i_{k+1})$. Since $L'$ is
$k$-thin, $\ell'$ must intersect some segment $\ell(i^{*})\in L'$.
The corresponding task $i^{*}$ then induces an m-task $m\in M(L')$
with $b(i_{1})\leq b(m)\leq b(i_{k+1})<b_{max}$. Hence $(T(L'),M(L'))$
is $(k+1/\delta)$-thin.

For the feasibility of $T(L')\cup M(L')$ recall that $T^{*}$ is
feasible and all tasks in $T^{*}$ are $\delta$-large. Hence, on
every edge $e$ each m-task $m_{e}$ uses at most as much capacity
as the tasks from $T^{*}$ whose segments are in $L'_{e}$ (the latter tasks in a
sense were replaced by $m_e$).
\end{proof}
Next we reduce the problem of finding a $k$-thin set $L'$ with low
weight to the case that each segment $\ell(i)$ starts at $e(i)$ and
either goes only to the right or only to the left.
See Figures~\ref{fig:maze-construction-initial}
and~\ref{fig:maze-construction-split}.
Formally, we split
each segment $\ell(i)$ into two segments $\ell_{L}(i)$ and $\ell_{R}(i)$
such that $\ell_{L}(i)$ contains the edges of $P(i)$ between $s(i)$
and the right vertex of the bottleneck edge $e(i)$ and symmetrically
for $\ell_{R}(i)$. So $\ell_{L}(i)$ and $\ell_{R}(i)$ overlap on
$e(i)$. We set $w(\ell_{L}(i))=w(\ell_{R}(i))=w(i)$. We define
$L_{L}:=\{\ell_{L}(i): \ell(i)\in L\}$ and $L_{R}:=\{\ell_{R}(i): \ell(i)\in L\}$.
The next lemma shows that it suffices to find low weight $k$-thin
sets for $L_L$ and $L_R$.
\begin{lem}
\label{lem:LR-split}
Given $k$-thin sets $L'_{L}$ for $L_{L}$
and $L'_{R}$ for $L_{R}$, there is a $2k$-thin set $L'$ for $L$
with $w(L')\le w(L'_{L})+w(L'_{R})$. 
\end{lem}
\begin{proof}
We add a line segment $\ell(i)$ to $L'$ if and only if $\ell_L(i)\in L'_{L}$
or $\ell_R(i)\in L'_{R}$. It follows directly that $w(L')\le
w(L'_{L})+w(L'_{R})$.
Now any vertical segment $\ell'$ crossing at least $2k+1$ segments in $L'$ must either cross
$k+1$ segments from $L_{L}$ or $k+1$ segments from $L_{R}$. Thus, $\ell'$
crosses a segment in $L_{L}'$ or a segment in $L_{R}'$, and hence $\ell'$ crosses a segment in $L'$. 
\end{proof}
Consider now only the segments $L_{R}$ (a symmetric argument holds for $L_{L}$). The next step is to reduce the problem to the case
where, intuitively
speaking, the edge capacities are strictly increasing and all segments
contain the leftmost edge of the graph. To simplify the description of
the step after this one, we also enforce that new segments have
different $y$-coordinates. Formally,
let us assume that task labels $i$ are integers between
$1$ and $\Card{T}$
(in any order). For each $\ell_R(i)=(v,u)\times \{b(i)\}\in L_R$, we
construct a segment $(1,u)\times \{b(i)+M\cdot v+\eps\cdot i\}$, which we
denote by $\tilde{\ell}_R(i)$. Here $M:=1+\max_{e}\{u_{e}\}$ and
$\eps=\frac{1}{\Card{T}+1}$ (so that $\eps\cdot i<1$). Define
$\tilde{L}_R:=\{\tilde{\ell}_R(i):i\in T^*\}$ and
$w(\tilde{\ell}_R(i))=w(i)$. 
(See Figure~\ref{fig:maze-construction-shift}.)
\begin{lem}\label{lem:left-extend-perturb}
Given a $k$-thin set $\tilde{L}'_R$ for $\tilde{L}_R$, there is a
$k$-thin set $L'_R$ for $L_R$ with $w(\tilde{L}'_R)= w(L'_R)$.
A symmetric claim holds for $\tilde{L}_L$ and $L_L$.
\end{lem}
\begin{proof}
We prove the first claim only, the proof of the second one being symmetric.
Let $L'_R:=\{\ell_R(i)\in L_R: \tilde{\ell}_R(i)\in \tilde{L}'_R\}$.
Clearly $w(L'_R)=w(\tilde{L}'_R)$. Consider any vertical segment
$\{x\}\times (y_b,y_t)$ that intersects at least $k+1$ segments from
$L_R$. Let $\ell_R(i_1),\ldots\ell_R(i_{k+1})$ be $k+1$ such segments of
lowest capacity, breaking ties according to the lowest label $i$ of the
corresponding tasks. To prove the lemma it suffices to show that at
least one such segment $\ell_R(i_{j^*})$ belongs to $L'_R$.

W.l.o.g., assume that for any $1\leq j\leq k$, $s(i_j)$ is equal to or to the left of $s(i_{j+1})$,
and $i_j<i_{j+1}$ if $s(i_j)=s(i_{j+1})$.
Then by construction
$\tilde{y}_1<\ldots <\tilde{y}_{k+1}$, where $\tilde{y}_j$ is the
$y$-coordinate of segment $\tilde{\ell}_R(i_j)$. Consider a vertical
segment $\{x\}\times (\tilde{y}_1-\eps,\tilde{y}_{k+1}+\eps)$. For $\eps>0$ small enough, we can assume that $\ell'$ 
intersects precisely the segments
$\tilde{\ell}_R(i_1),\ldots,\tilde{\ell}_R(i_{k+1})$. Hence
$\tilde{\ell}_R(i_{j^*})\in \tilde{L}'_R$ for some $1\leq j^*\leq k+1$.
If follows that $\ell_R(i_{j^*})\in L'_R$ as required.
\end{proof}

It remains to prove that there is a $k$-thin set for $\tilde{L}_{R}$
whose weight is bounded by $O(\frac{1}{k})w(\tilde{L}_{R})$. We do
this by reducing this problem to a min-flow problem in a directed
tree network. 

Let $k\in\mathbb{N}$ be even. We consider the following hierarchical
decomposition of the segments in $\tilde{L}_{R}$, which corresponds
to a (directed) rooted out-tree $\cD$ (see
Figures~\ref{fig:maze-construction-decomposition}
and~\ref{fig:maze-construction-tree}). We construct $\cD$ iteratively,
starting from the root. Each node $w$ of $\cD$ is labelled with a triple
$(e_{w},I_{w},R_{w})$, where $e_{w}$ is an edge in $E$,
$I_{w}\subseteq[0,\infty)$
is an interval, and $R_{w}$ contains all segments that contain $e$
and whose $y$-coordinate is in $I_{w}$ (the \emph{representative
segments} of $w$). Let $e_{r}\in E$ be the rightmost edge that is
contained in at least $k-1$ segments. We let the root $r$ of $\cD$
be labelled with $(e_{r},[0,\infty),R_{r})$. For any constructed
node $w$, if $e_{w}$ is the leftmost edge of the graph, then $w$
is a leaf. Otherwise, consider the edge $e'$ to the left of $e_{w}$,
and let $R'$ be the segments in $I_{w}$ that contain $e'$. Note
that, by the initial preprocessing of the instance, each edge can
be the rightmost edge of at most one segment (task), hence
$\Card{R'}\leq\Card{R_{w}}+1$.
If $\Card{R'}<k$, we append to $w$ a child $w'$ (with a directed arc
$(w,w')$) with label $(e',I_{w},R')$.
Otherwise (i.e., if $\Card{R'}=k$), we append to $w$ two children
$w_{b}$ and $w_{t}$, which are labelled as follows. Let
$\tilde{\ell}_R(i_{1}),...,\tilde{\ell}_R(i_{k})$ be the
segments in $R'$, sorted increasingly by $y$-coordinate. We partition
$R'$ into $R_{b}=\{\tilde{\ell}_R(i_{1}),...,\tilde{\ell}_R(i_{k/2})\}$ and $R_{t}=\{\tilde{\ell}_R(i_{k/2+1}),...,\tilde{\ell}_R(i_{k})\}$.
Let $y_{0}$ be a value such that all segments in $R_{b}$ have a $y$-coordinate
strictly less than $y_{0}$. We label $w_{b}$ and $w_{t}$ with $(e',I_{w}\cap[0,y_0),R_{b})$
and $(e',I_{w}\cap[y_0,\infty),R_{t})$, respectively.

Consider a given segment $\tilde{\ell}_R(i)$, and the nodes $w$ of $\cD$ that
have $\tilde{\ell}_R(i)$ as one of their representative segments $R_{w}$.
Then the latter nodes induce a directed path $\cD(i)$ in $\cD$.
To see this, observe that if $\tilde{\ell}_R(i)\in R_{w}$, then either $w$
is a leaf or $\tilde{\ell}_R(i)\in R_{w'}$ for exactly one child $w'$ of $w$.
Furthermore, each $\tilde{\ell}_R(i)$ belongs to $R_{w}$ for some leaf $w$
of $\cD$ (i.e., no $\cD(i)$ is empty). 

We call a set of segments $\tilde{L}'_{R}\subseteq\tilde{L}_{R}$ a \emph{segment
cover }if for each node $w$ of $\cD$ it holds that $R_{w}\cap\tilde{L}'_{R}\ne\emptyset$. 
\begin{lem}
\label{lem:thin-lines}If $\tilde{L}'_{R}\subseteq\tilde{L}_{R}$
is a segment cover then $\tilde{L}'_{R}$ is $2k$-thin for $\tilde{L}_{R}$. A symmetric claim holds for $\tilde{L}_R$.
\end{lem}
\begin{proof}
We prove the first claim only, the proof of the second one being symmetric.
Consider any vertical segment $\ell'=\{x\}\times (y_b,y_t)$ crossing at least $2k+1$ segments from $\tilde{L}_{R}$, and let 
$\tilde{L}''$ be $2k+1$ such segments of lowest $y$-coordinate. 
Let also $e=(u,u+1)$ be the edge such that $x\in (u,u+1)$, and $\tilde{\ell}_R(i_1),\ldots,\tilde{\ell}_R(i_h)$ be the segments containing edge $e$ in increasing order of $y$ coordinate. Observe that segments $\tilde{L}''$ induce a subsequence $\tilde{\ell}_R(i_{j}),\ldots,\tilde{\ell}_R(i_{j+2k})$ of $\tilde{\ell}_R(i_1),\ldots,\tilde{\ell}_R(i_h)$. Furthermore, the representative sets $R_w$ of nodes $w$ such that $e_w=e$ partition $\tilde{\ell}_R(i_1),\ldots,\tilde{\ell}_R(i_h)$ into subsequences, each one containing between $k/2$ and $k-1$ segments. It follows that there must be one node $w'$ such that $R_{w'}\subseteq \{\tilde{\ell}_R(i_{j}),\ldots,\tilde{\ell}_R(i_{j+2k})\}$. Since $\tilde{L}'_R\cap R_{w'}\neq \emptyset$ by assumption, it follows that $\tilde{\ell}_R(i_{j^*})\in \tilde{L}'_R$ for some $j\leq j^*\leq j+2k$.
\end{proof}

It remains to show that there is a segment cover with small weight.
\begin{lem}
\label{lem:fractionalDirectedPaths} There exists a segment cover $\tilde{L}'_{R}\subseteq\tilde{L}_{R}$
with $w(\tilde{L}'_{R})\leq\frac{2}{k}\cdot w(\tilde{L}_{R})$ (where
$k$ is the parameter used in the construction of $\cD$). \end{lem}
\begin{proof}
We can formulate the problem of finding a $\tilde{L}'_{R}$ satisfying
the claim as a flow problem. We augment $\cD$ by appending a dummy
node $w'$ to each leaf node $w$ with a directed edge $(w,w')$ (so
that all the original nodes are internal) and extend the paths $\cD(i)$
consequently (so that each path contains exactly one new edge $(w,w')$).

We define a min-flow problem, specified by a linear program. For each
directed path $\cD(i)$ we define a variable $x_{i}\in[0,1]$.
Let $A$ denote the set of all arcs in $\cD$. For each arc $a$ denote
by $T_{a}$ all values $i$ such that $\cD(i)$ uses $a$. We solve
the following LP:

\begin{align*}
\min &\sum_{i:\ell(i)\in\tilde{L}_{R}}w(i)\cdot x_{i} & \\
\textrm{s.t.} &\quad \sum_{i\in T_{a}}x_{i}  \ge 1  &&\forall a\in A\\
     &\quad x_{i} \geq 0   && \forall\ell(i)\in\tilde{L}_{R}.
\end{align*}

By the construction of $\cD$ every arc is used by at least $k/2$ paths.
Hence, the linear program has a fractional solution of weight
$\sum_{i}w(i)\cdot\frac{2}{k}=\frac{2}{k}\cdot w(\tilde{L}_{R})$,
which is obtained by setting $x_{i}:=2/k$ for each $i$. Since the underlying
network $\cD$ is a directed tree and all paths follow the direction 
of the arcs, the resulting network flow matrix
is totally unimodular, see~\cite{Schrijver03}. Therefore, there exists
also an integral solution with at most the same weight. This integral
solution induces the set $\tilde{L}'_{R}$. 
\end{proof}
Now the proof of Lemma~\ref{lem:apx} follows from the previous reductions.
\begin{proof}[Proof of Lemma~\ref{lem:apx}]
Suppose we are given the optimal solution $T^{*}$. As described above, we construct
the sets $L$, $L_{L}$, $L_{R}$, $\tilde{L}_{L}$, and $\tilde{L}_{R}$.
We compute segment covers $\tilde{L}'_{L}$ for $\tilde{L}_{L}$ and
$\tilde{L}'_{R}$ for $\tilde{L}_{R}$ as described in the proof of Lemma~\ref{lem:fractionalDirectedPaths}.
By Lemma~\ref{lem:thin-lines} they are $2k$-thin for $\tilde{L}_{L}$
and $\tilde{L}_{R}$, respectively. By Lemma~\ref{lem:left-extend-perturb}
we obtain $2k$-thin sets $L'_{L}$ and $L'_{R}$ for $L_{L}$ and
$L_{R}$, respectively, with $w(L'_{L})=w(\tilde{L}'_{L})$ and
$w(L'_{R})=w(\tilde{L}'_{R})$. By Lemma~\ref{lem:LR-split} this yields a
$4k$-thin
set $L'$ for $L$ whose weight is bounded by $w(L'_{L})+w(L'_{R})$.
Finally, set $(\tilde{T},\tilde{M}):=(T(L'),M(L'))$. This maze pair is
feasible by definition.  Furthermore, by Lemma~\ref{lem:segments}, it is
$(4k+\frac{1}{\delta})$-thin and its  weight is bounded by
$w(\tilde{T})\leq w(\tilde{L}'_{L})+w(\tilde{L}'_{R})\le\frac{2}{k}\cdot(w(\tilde{L}_{L})+w(\tilde{L}_{R}))\le\frac{4}{k}\cdot
w(L)=\frac{4}{k}\cdot w(T^{*})$.
\end{proof}

\section{The Dynamic Program}\label{sec:DP}

In this section we present a dynamic program computing the weakly feasible
$k$-thin maze pair with maximum profit (where $k=k(\eps,\delta)$ will
correspond to the constant $k$ of Lemma~\ref{lem:apx}).
Thus, we prove Lemma \ref{lem:DP}.

Let $k\in \mathbb{N}$. To simplify the description and analysis of our DP, we introduce
the following assumptions and notations. For having a clearly defined
root in the DP, we add an edge $e^{*}$ to the left of $E$ with $u_{e^{*}}=0$
(note that $e^{*}$ is used by no task). 
For notational convenience, we add to $M$ two special dummy
m-tasks $\bot$ and $\top$. The paths of $\bot$ and $\top$ span all the
edges of the graph, and they both have demand zero. Furthermore,
$b(\top):=+\infty$
and $b(\bot):=0$. In particular, with these definitions we have that
$\abv_{e}(\top,T')=\crit_{e}(\top,T')=\emptyset$,
$\abv_{e}(\bot,T')=T'\cap T_{e}$, and $\crit_{e}(\bot,T')=\emptyset$.
We let $e(\bot)$ be the rightmost edge of the graph, and we leave
$e(\top)$ unspecified. However, when talking about weak-feasibility
and $k$-thinness of a maze pair $(T',M')$ we will ignore dummy tasks,
that is, we will implicitly consider $(T',M'-\{\bot,\top\})$. 

For any $e\in E$, $T'\subseteq T$, and any two m-tasks $m'$ and
$m''$ with $b(m')<b(m'')$, the \emph{boundary tasks} in $T'$ for the
triple $(e,m',m'')$ are the tasks $\bound_{e}(m',m'',T'):=\{i\in T'\cap
T_{e}\,:\, b(m')<b(i)\leq b(m'')\}$.
Intuitively, boundary tasks $i$ are the tasks using edge
$e$ such that the segment corresponding to~$i$ is sandwiched between
the segments corresponding to $m'$ and $m''$.

In our dynamic programming table we introduce a cell for each entry
of the form $c=(e,m_{\up},C_{\up},m_{\dn},C_{\dn},B)$ where: 
\begin{itemize}\itemsep0pt 
\item $e$ is an edge; 
\item $m_{\dn}\in M_{e}$ and $m_{\up}\in M_{e}$, $b(m_{\dn})<b(m_{\up})$; 
\item $C_{\dn}\subseteq\crit(m_{\dn},T)$ and $C_{\up}\subseteq\crit(m_{\up},T)$,
with $\Card{C_{\up}},\Card{C_{\dn}}\leq ncrit(\delta)$; 
\item $B\subseteq\bound_{e}(m_{\dn},m_{\up},T)$, with $\Card{B}\leq k$. 
\end{itemize}

Observe that $C_{\dn}$ and $B$ are disjoint, while $C_{\up}$ might
overlap with both $C_{\dn}$ and $B$. 
For such a cell to exist we further impose the following \emph{consistency property}: 
\begin{itemize}
\item $(B\cup C_{\dn}\cup C_{\up},\{m_{\dn},m_{\up}\})$ is weakly feasible;
\item for $T'=B\cup C_{\dn}\cup C_{\up}$ we require $\crit(m_{\dn},T')=C_{\dn}$, $\crit(m_{\up},T')=C_{\up}$, and $\bound_{e}(m_{\dn},m_{\up},T')=B$.
\end{itemize}

Given a DP cell $c=(e,m_{\up},C_{\up},m_{\dn},C_{\dn},B)$, as a shorthand notation we use $e(c):=e$, $m_{\up}(c):=m_{\up}$ and similarly for the
other entries of the cell.
We also define $e_{\dn}=e_{\dn}(c):=e(m_{\dn})$ and
$e_{\up}=e_{\up}(c):=e(m_{\up})$ (we set $e_{\up}=e$ if $m_{\up}=\top$).

The idea behind a cell $c$ is as follows. We define $E(c)$ as the
set of edges between $e_{\up}$ (included) and $e_{\dn}$ (excluded) (if
$e_{\up}=e_{\dn}$, we assume $E(c)=\emptyset$). We define $T(c)$
as the set of tasks $i$ with bottleneck edge in $E(c)$ such
that $b(i)>b(m_{\up})$ or $P(i)$ {contains neither $e$
nor $e_{\up}$. We define $M(c)$ similarly w.r.t. m-tasks. For a geometric
intuition we can think of cell $c$ as defining an area such that $T(c)$ and
$M(c)$ belong entirely inside---see Figure~\ref{fig:area}.

Our goal is to compute the maze-pair $(T_{c},M_{c})$ with $T_{c}\subseteq T(c)$
and $M_{c}\subseteq M(c)$ with maximum weight $w(c):=w(T_{c})$ such that: 
\begin{itemize}\itemsep0pt
\item $(T_{c}\cup B\cup C_{\dn}\cup C_{\up},M_{c}\cup\{m_{\dn},m_{\up}\})$
is weakly feasible; 
\item $(T_{c}\cup B,M_{c}\cup\{m_{\dn},m_{\up}\})$ is $k$-thin; 
\item If $i\in crit(m_{\up},T_{c})$ then $i\in C_{\up}$ (\emph{inclusion
property}). 
\end{itemize}
We call maze-pairs fulfilling the above properties \emph{feasible
for~$c$.} From this definition it follows that the optimal solution
for the cell $c^{*}:=(e^{*},\bot,\emptyset,\top,\emptyset,\emptyset)$
is the weakly feasible $k$-thin maze pair $(T_{c^{*}},M_{c^{*}})$
with maximum weight $w(T_{c^{*}})$.

We define a partial order $\prec$ for the cells and fill in the DP-table 
w.r.t. this order (breaking ties arbitrarily). Intuitively speaking,
we define $\prec$ to ensure that $c'\prec c''$ if the area (within the capacity curve) corresponding to $c'$ is contained in the area corresponding to $c''$. The
following definition achieves this: for two edges $e'$ and $e''$,
we let $\Card{e'-e''}$ be the number of edges between $e'$ and $e''$,
boundary included. We define that $c'\prec c''$ if (in a lexicographic
sense)
$(\Card{e_{\up}(c')-e_{\dn}(c')},\Card{e(c')-e_{\dn}(c')})$ $<_{lex}$\newline $(\Card{e_{\up}(c'')-e_{\dn}(c'')},\Card{e(c'')-e_{\dn}(c'')})$.

The base case cells are obtained when $e=e_{\dn}$. In this case one
must have $m_{\up}=\top$, and hence $e_{\up}=e$. Also $T(c)=\emptyset=M(c)$.
For those cells we set $(T_{c},M_{c}):=(\emptyset,\emptyset)$ (hence
$w(c)=0$).

Consider a cell $c$ that is not a base case. For the sake of presentation,
assume that $e_{\dn}$ is to the right of $e$, the other case being
symmetric. Let $e_{r}$ be the first edge to the right of $e$ (possibly
$e_{r}=e_{\dn}$). We will compute $(T_c,M_c)$ as a
function of some pairs $(T_{c'},M_{c'})$ with $c'\prec c$,
considering the following three branching cases (see
Figure~\ref{fig:branching}):

\begin{figure}
\centering
\scalebox{.4}{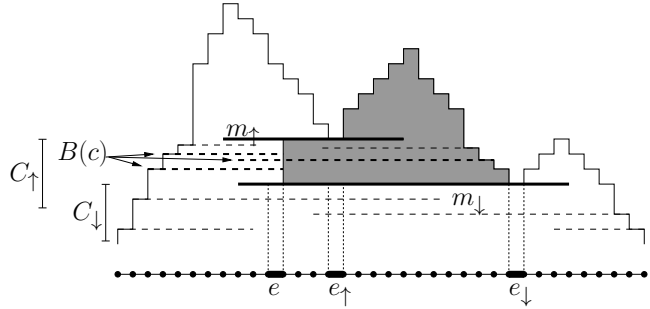}
\caption{\label{fig:area}
Tasks $B\cup C_\dn \cup C_\up$ (dashed) and area associated to a DP
cell~$c$. Tasks in $C_\up$ ($C_\dn$) use a common edge with $m_\up$ ($m_\dn$).
Tasks (resp., m-tasks) that lie entirely within the shaded area
are those that belong to $T(c)$ (resp., $M(c)$).}
\end{figure}

\begin{figure}
\centering
\subfigure[\label{fig:cases-single}single branching 
]{\scalebox{.45}{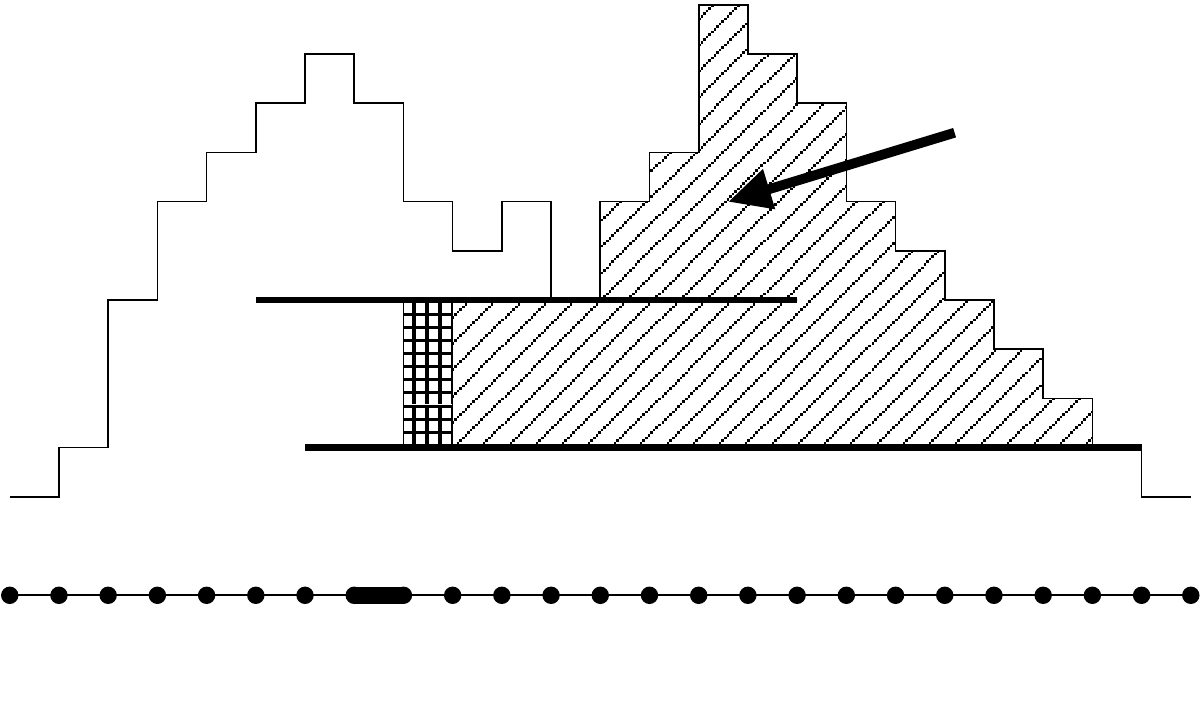}}\hspace{.3cm}
\subfigure[\label{fig:cases-topBottom}top-bottom branching]{\scalebox{.45}{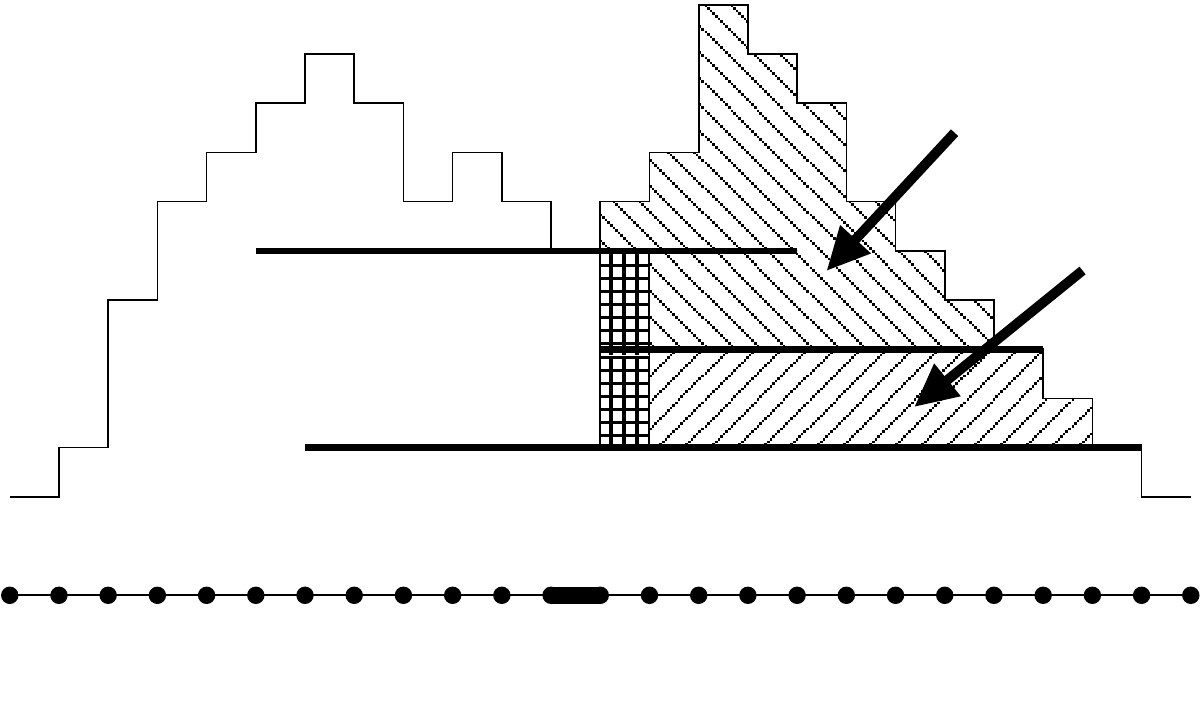}}\hspace{.5cm}
\subfigure[\label{fig:cases-leftRight}left-right branching]{\scalebox{.45}{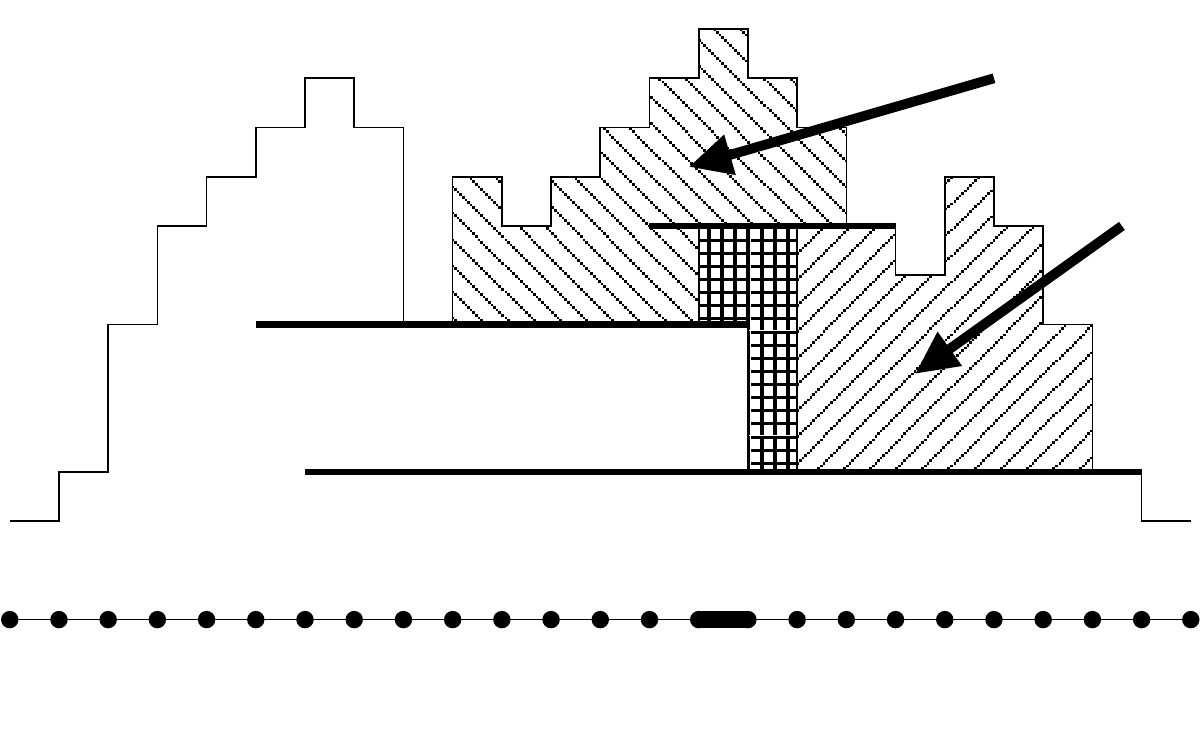}}\hspace{.5cm}
\caption{The three branching cases. The area of the cell $c$ is the area of
the subcell(s) $c_s$, $(c_t,c_b)$, and $(c_l,c_r)$, respectively, in addition to the checkered pattern.}
\label{fig:branching}
\end{figure}

\begin{itemize}
\item \textbf{(single branching)} This case applies only when $m_{\up}$
uses both $e$ and $e_{r}$ (possibly $m_{\up}=\top$). Consider
any feasible entry $c_{s}=(e_{r},m_{\dn},C_{\dn},m_{\up},C_{\up},B_{s})$
with the following extra \emph{compatibility property}: 

\begin{quote}
for $T':=C_{\dn}\cup C_{\up}\cup B\cup B_{s}$, one has $\crit(m_{\dn},T')=C_{\dn}$,
$\crit(m_{\up},T')=C_{\up}$, $\bound_{e}(m_{\dn},m_{\up},T')=B$,
and $\bound_{e_{r}}(m_{\dn},m_{\up},T')=B_{s}$. 
\end{quote}

Set $w_{sb}(c)\leftarrow\max_{c_{s}}\{w(c_{s})+w(B_{s} \setminus B)\}$. 

\item \textbf{(top-bottom branching)} This case applies only when $m_{\up}$
uses both $e$ and $e_{r}$ (possibly $m_{\up}=\top$). Consider
any m-task $m_{mid}\neq\top$ that has $e_{r}$ as its leftmost edge
and such that $b(m_{\dn})<b(m_{mid})<b(m_{\up})$. Consider any pair
of feasible entries\\
$c_{b}=(e_{r},m_{\dn},C_{\dn},m_{mid},C_{mid},B_{b})$ and\\
$c_{t}=(e_{r},m_{mid},C_{mid},m_{\up},C_{\up},B_{t})$ with the following
extra \emph{compatibility property}: 

\begin{quote}
for $T':=C_{\dn}\cup C_{\up}\cup C_{mid}\cup B\cup B_{b}\cup B_{t}$,
one has $\crit(m_{\up},T')=C_{\up}$, $\crit(m_{\dn},T')=C_{\dn}$,
$\crit(m_{mid},T')=C_{mid}$, $\bound_{e}(m_{\dn},m_{\up},T')=B$,
$\bound_{e_{r}}(m_{\dn},m_{mid},T')=B_{b}$, and\\
$\bound_{e_{r}}(m_{mid},m_{\up},T')=B_{t}$. 
\end{quote}

Set $w_{tb}(c)\leftarrow\max_{(c_{b},c_{t})}\{w(c_{b})+w(c_{t})+w((B_{b}\cup B_{t})\setminus B)\}$. 

\item \textbf{(left-right branching)} This branching applies only to the
case that $e$ is the rightmost edge of $m_{\up}$, and $m_{\up}\neq\top$.
Consider any m-task $m_{abv}$ that uses
both $e_{r}$ and $e$ and
with $b(m_{abv})>b(m_{\up})$ (possibly $m_{abv}=\top$). Consider the pairs of feasible entries\\
 $c_{l}=(e,m_{\up},C_{\up},m_{abv},C_{abv},B_{l})$ and\\
 $c_{r}=(e_{r},m_{\dn},C_{\dn},m_{abv},C_{abv},B_{r})$ with the following
extra \emph{compatibility property}: 

\begin{quote}
for $T':=C_{\dn}\cup C_{\up}\cup C_{abv}\cup B\cup B_{l}\cup B_{r}$,
one has $\crit(m_{\up},T')=C_{\up}$, $\crit(m_{\dn},T')=C_{\dn}$,
$\crit(m_{abv},T')=C_{abv}$, $\bound_{e}(m_{\dn},m_{\up},T')=B$,
$\bound_{e_{r}}(m_{\dn},m_{abv},T')=B_{r}$, and\\
$\bound_{e}(m_{\up},m_{abv},T')=B_{l}$. 
\end{quote}

We set $w_{lr}(c)\leftarrow\max_{(c_{l},c_{r})}\{w(c_{l})+w(c_{r})+w((B_{l}\cup B_{r}) \setminus B)\}$. 

\end{itemize}
Finally, we define $w(c):=\max\{w_{sb}(c),w_{tb}(c),w_{lr}(c)\}$.
Depending on the case attaining the maximum, we define $(T_{c},M_{c})$:
if the maximum is achieved in the single-branching case for some $c_{s}$,
then we set $T_{c}\leftarrow T_{c_{s}}\cup(B_{s}\setminus B)$ and $M_{c}\leftarrow M_{c_{s}}$.
If the maximum is achieved in the top-bottom branching for some $c_{b}$
and $c_{t}$, we set $T_{c}\leftarrow T_{c_{b}}\cup T_{c_{t}}\cup((B_{b}\cup B_{t}) \setminus B)$
and $M_{c}\leftarrow M_{c_{b}}\cup M_{c_{t}}\cup\{m_{mid}\}$. Similarly,
if the maximum is achieved in the left-right branching for some $c_{l}$
and $c_{r}$, we set $T_{c}\leftarrow T_{c_{l}}\cup T_{c_{r}}\cup((B_{l}\cup B_{r}) \setminus B)$
and $M_{c}\leftarrow M_{c_{l}}\cup M_{c_{r}}\cup\{m_{abv}\}$.

Observe that, in the single branching case, one has
that $|e_\up(c_s)-e_\dn(c_s)|=|e_\up(c)-e_\dn(c)|$ and that
  $|e(c_s)-e_\dn(c_s)|<|e(c)-e_\dn(c)|$. In the other cases one has
  $\Card{e_\up(c')-e_\dn(c')}<\Card{e_\up(c)-e_\dn(c)}$, where $c'\in
  \{c_b,c_t,c_l,c_r\}$. Hence $c_s,c_b,c_t,c_l,c_r\prec c$ as required.
Note also that $c^{*}$ is the only feasible table entry associated to
edge $e^{*}$ and for any other entry $c$ it holds that $c\prec c^{*}$.
The DP outputs $(T_{c^{*}},M_{c^{*}})$ and we return $T_{c^{*}}$
as the computed set of tasks.
\begin{lem}\label{lem:DP-running-time}
For any constant $\delta,k>0$, $k\in \mathbb{N}$, the above dynamic program runs in polynomial time.
\end{lem}
\begin{proof} Note that the number
of cells is polynomially bounded since $\Card{T\cup M}$ is polynomially
bounded in $\Card{T}$ and $k(\eps,\delta)$ and $ncrit(\delta)$ are constants.
Similarly, in the computation of each pair $(T_{c},M_{c})$ one has
to consider only a polynomial number of possibilities. Altogether,
the dynamic program runs in polynomial time. \end{proof}

We next show the correctness of the dynamic program. Consider any
cell $c$. First observe that $T_c\subseteq T(c)$ and $M_c\subseteq M(c)$. Also, by an easy induction, any two distinct m-tasks in $M_c$ have different bottleneck capacity. In other terms, $(T_c,M_c)$ is a well-defined maze pair. 
We next prove that $w(T_{c})\geq w(T'_{c})$ for any feasible
pair $(T'_{c},M'_{c})$ for $c$ (Lemma~\ref{lem:DP-optimality}).
Then we prove that $(T_{c},M_{c})$ is feasible for~$c$, by showing
that it has the $3$ required properties (Lemmas \ref{lem:DP-k-thinness},
\ref{lem:DP-inclusion}, \ref{lem:DP-weak-feasibility}).

For showing the next lemma, we prove that if a pair $(T'_{c},M'_{c})$
is feasible for a cell $c$, then it can be decomposed into the feasible
solution for a cell $c_{s}$ and the tasks in $B(c_{s})\setminus  B$ or into
feasible solutions for two cells $c_{t},c_{b}$ (or $c_{l},c_{r}$)
and the tasks in $(B(c_{t})\cup B(c_{b}))\setminus  B$ (or $(B(c_{l})\cup B(c_{r})) \setminus B$),
depending on the applying branching case.

\begin{lem}\label{lem:DP-optimality} Let $(T'_{c},M'_{c})$ be a feasible maze-pair for cell $c=(e,m_{\dn},C_{\dn},m_{\up},C_{\up},B)$.
Then $w(T'_{c})\le w(T_{c})$. 
\end{lem} 
\begin{proof} We show the
claim by induction, following the partial order $\prec$ on cells. For the base cases, it is clear that $w(T'_{c})=w(T_{c})=0$
and the claim follows.

Now consider a non-base-case cell $c$ and suppose the
claim is true for all cells $c'$ with $c'\prec c$. W.l.o.g.~assume
again that $e_{\dn}$ lies on the right of $e$, and let $e_{r}$
be the edge adjacent to $e$ on the right. We distinguish cases, depending
on which m-tasks use $e_{r}$. 

First suppose that there is no m-task $m_{mid}\in M'_{c}\cap M_{e_{r}}$
with $b(m_{\dn})<b(m_{mid})<b(m_{\up})$ using $e_{r}$ and that $m_{\up}$
uses $e_{r}$, where possibly $m_{\up}=\top$ (single branching case).
Then consider the DP-cell $c_{s}=(e_{r},m_{\dn},C_{\dn},m_{\up},C_{\up},B_{s})$
with $B_{s}=\bound_{e_{r}}(m_{\dn},m_{\up},T'_{c}\cup B)$. 
Since $(T'_c,M'_c)$ is feasible for $c$, $c_s$ is indeed a cell in our DP-table. In particular, observe that $|B_s|\le k$ since $(T'_c,M'_c)$ is $k$-thin. The consistency property follows by the weak feasibility of $(T'_c,M'_c)$ and from the compatibility property of the single branching.
By induction, we know that the DP computed the optimal solution
$(T_{c_{s}},M_{c_{s}})$ for $c_{s}$. In particular, $w(T_{c_{s}})\ge w(T'_{c})- w(B_{s}\setminus B)$
since $(T'_{c}\setminus (B_{s}\setminus B),M'_{c})$ is feasible for~$c_{s}$. 
By definition of the DP-transition,
\[\begin{split}
w(T_{c})&\ge w_{sb}(c)\ge w(T_{c_{s}})+w(B_{s}\setminus B)\\
&\ge(w(T'_{c})-w(B_{s}\setminus B))+w(B_{s}\setminus B)\\
&=w(T'_{c}).
\end{split}\]

Next consider the case that there is an m-task $m_{mid}\in M'_{c}\cap M_{e_{r}}$
with $b(m_{\dn})<b(m_{mid})<b(m_{\up})$ using $e_{r}$. Note that by our preprocessing then $m_{\up}$
uses $e_{r}$ where
possibly $m_{\up}=\top$ (top-bottom branching).
Also observe that there can be at most one such
task $m_{mid}$ by our preprocessing and using that any two m-tasks in a maze pair have different bottleneck capacities.
Let us consider the (bottom) cell $c_{b}=(e_{r},m_{\dn},C_{\dn},m_{mid},C_{mid},B_{b})$
and the (top) cell $c_{t}=(e_{r},m_{mid},C_{mid},m_{\up},C_{\up},B_{t})$
where we define $B_{b}:=\bound_{e_{r}}(m_{\dn},m_{mid},T'_{c}\cup B)$, $B_{t}:=\bound_{e_{r}}(m_{mid},m_{\up},T'_{c}\cup B)$, and $C_{mid}:=\crit(m_{mid},T'_{c}\cup B)$.
Also in this case, the feasibility of $(T'_c,M'_c)$ for $c$
implies that $c_l$ and $c_r$ are in fact DP-cells. In particular, since $(T'_{c},M'_{c})$ is weakly feasible, $|C_{mid}|\leq ncrit(\delta)$ by Lemma~\ref{lem:critical}.
The pair $(T'_{c}\cap T(c_{b}),M'_{c}\cap M(c_b)\})$
is feasible for $c_{b}$ and the pair $(T'_{c}\cap T(c_{t}),M'_{c}\cap M(c_b))$
is feasible for~$c_{t}$. In this case $T'_c$ is partitioned by $T'_{c}\cap T(c_{b})$, $T'_{c}\cap T(c_{t})$, and 
$(B_{b}\cup B_{t})\setminus B$. Hence,
\[\begin{split}
w(T_{c})&\ge w_{tb}(c) \ge  w(c_{b})+w(c_{t})+w((B_{b}\cup B_{t})\setminus B)\\
&\ge w(T'_{c}\cap T(c_{b}))+w(T'_{c}\cap T(c_{t}))+w((B_{b}\cup B_{t})\setminus B)\\
&= w(T'_{c}).
\end{split}\]

Finally, consider the case that there is no m-task $m_{mid}\in M'_{c}\cap M_{e_{r}}$
with $b(m_{\dn})<b(m_{mid})<b(m_{\up})$ and that $m_{\up}$ does
not use $e_{r}$ (left-right branching case). Let $m_{abv}\in M'_{c}\cap M_{e_{r}}$
be the m-task minimizing $b(m_{abv})$ such that $b(m_{abv})>b(m_{\up})$
(possibly $m_{abv}=\top$). 
By the preprocessing 
of the input tasks, if $m_{abv}\neq\top$, then $m_{abv}$ must use
$e$, as well (otherwise two m-tasks with different bottleneck capacities would share one endpoint). 
Consider the DP-cells $c_{l}=(e,m_{\up},C_{\up},m_{abv},C_{abv},B_{l})$ and
$c_{r}=(e_{r},m_{\dn},C_{\dn},m_{abv},C_{abv},B_{r})$ where we define $B_{l}=\bound_{e}(m_{\up},m_{abv},T'_{c}\cup B)$,
$B_{r}=\bound_{e_{r}}(m_{\dn},m_{abv},T'_{c}\cup B)$, and $C_{abv}=\crit(m_{abv},T'_{c}\cup B)$.

Again, since $(T'_{c},M'_{c})$ is feasble for $c$, $c_l$ and $c_r$ are in fact DP-cells.

Also, the pair $(T'_{c}\cap T(c_{l}),M'_{c}\cap M(c_l))$
is feasible for $c_{l}$ and the pair $(T'_{c}\cap T(c_r),M'_{c}\cap M(c_r))$
is feasible for~$c_{r}$. By induction, we know that the DP computed
the optimal solutions $(T_{c_{l}},M_{c_{l}})$ and $(T_{c_{r}},M_{c_{r}})$
for $c_{l}$ and $c_{r}$, respectively. Observe that $T'_c$ is partitioned by $T'_{c}\cap T(c_{l})$, $T'_{c}\cap T(c_{r})$, and 
$(B_{l}\cup B_{r})\setminus B$. Hence,
\[\begin{split}
w(T_{c}) &\ge w_{lr}(c)\ge w(c_{l})+w(c_{r})+w((B_{l}\cup B_{r})\setminus B)\\
&= w(T'_{c}\cap T(c_{l}))+w(T'_{c}\cap T(c_{r}))+w((B_{l}\cup B_{r})\setminus B)\\
&= w(T'_{c}).
\end{split}\]
This concludes the proof.
\end{proof}

The proofs of the next three lemmas use a similar inductive pattern.
We show that whenever we extend the solution for a cell $c_{s}$ or
combine the solutions for two cells $c_{t,}c_{b}$ or $c_{l},c_{r}$
to a solution for some cell $c$ according to the DP-transition, then the
new solution is $k$-thin (has the inclusion property, is weakly feasible)
assuming that the original cells $c_{s}$ or $c_{t,}c_{b}$ or $c_{l},c_{r}$
were $k$-thin (have the inclusion property, are weakly feasible).

\begin{lem}[$k$-thinness] \label{lem:DP-k-thinness} For each
table entry
 $c=(e,m_{\dn},C_{\dn},m_{\up},C_{\up},B)$, we have that $(T_{c}\cup B,M_{c}\cup\{m_{\dn},m_{\up}\})$
is $k$-thin.
\end{lem} 
\begin{proof} It is sufficient to show that,
given any edge $e$ and any two tasks $m',m''\in(M_{c}\cup\{m_{\dn},m_{\up}\})\cap M_{e}$,
with $b(m')<b(m'')$ and such that there is no $m'''\in(M_{c}\cup\{m_{\dn},m_{\up}\})\cap M_{e}$
with $b(m')<b(m''')<b(m'')$, then the number of tasks $i$ in $(T_{c}\cup B)\cap T_{e}$
with $b(m')< b(i)\leq b(m'')$ is at most $k$. In other
terms, $\Card{\bound_{e}(m',m'',T_{c}\cup B)}\leq k$.

We prove the latter claim by induction, following the partial order
$\prec$ on the cells. For the base cases, 
recall that for each DP-cell $c$ we required that $\Card{B(c)}\le k$.
Hence, in that case $(T_{c}\cup B,M_{c}\cup\{m_{\dn},m_{\up}\})=(B,\{m_{\dn},m_{\up}\})$
and the claim is trivially true.

Now consider a non-base-case DP-cell $c$ and suppose the
claim is true for all cells $c'$ with $c'\prec c$. Assume w.l.o.g.~that
$e_{\dn}$ lies on the right of $e$. We distinguish the three branching
cases and show that in each case the pair $(T_{c},M_{c})$ is $k$-thin.

First suppose that the single branching case applies, that is, there
is a cell $c_{s}$ such that $T_{c}=T_{c_{s}}\cup(B(c_{s})\setminus B)$ and
$M_{c}=M_{c_{s}}$. By induction $(T_{c_{s}}\cup B(c_{s}),M_{c_{s}}\cup\{m_{\dn},m_{\up}\})$
is $k$-thin. Hence, it suffices to ensure that $\Card{\bound_{e}(m_{\dn},m_{\up},T_{c}\cup B)}\le k$.
However, the latter holds since
$\Card{\bound_{e}(m_{\dn},m_{\up},T_{c}\cup B)}=\Card{\bound_{e}(m_{\dn},m_{\up},B)}=\Card{B}$ by the compatibility property of the branching, and $|B|\leq k$ by the definition of DP-cells.

The same basic argument also works for the remaining two branching cases: it is sufficient to bound\newline $\Card{\bound_{e}(m_{\dn},m_{\up},T_{c}\cup B)}$, and an upper bound of $k$ follows from the compatibility property of the considered branching and by definition of DP-cells.
\end{proof}

\begin{lem}[Inclusion property]
\label{lem:DP-inclusion} For each table cell $c=(e,m_{\dn},C_{\dn},m_{\up},C_{\up},B)$, if $i\in crit(m_{\up},T_{c})$ then $i\in C_{\up}$.\end{lem}
\begin{proof}
We prove this claim by using the compatibility properties of the branching
procedures. The claim is trivially true for base case cells $c$ since $T_{c}\subseteq T(c)=\emptyset$.

Consider now a non-base-case cell $c$, and assume the claim holds
for any cell $c'\prec c$. Assume w.l.o.g.~that $e_{\dn}$ lies on
the right of $e$. Suppose that $T_{c}=T_{c_{s}}\cup(B_{s}\setminus B)$ for
some cell $c_{s}=(e_{r},m_{\dn},C_{\dn},m_{\up},C_{\up},B_{s})$
in the single branching case.
If $i\in crit(m_{\up},B_{s}\setminus B)$
then $i\in C_{\up}$ by the compatibility property of the single branching
procedure.
If $i\in crit(m_{\up},T_{c_{s}})$ then $i\in C_{\up}$
by the induction hypothesis. 

Assume now that $T_{c}=T_{c_{t}}\cup T_{c_{b}}\cup((B(c_{t})\cup B(c_{b}))\setminus B)$
for some cells $c_{b}=(e_{r},m_{\dn},C_{\dn},m_{mid},C_{mid},B_{b})$
and $c_{t}=(e_{r},m_{mid},C_{mid},m_{\up},C_{\up},B_{t})$ in the
top-bottom branching case. If $i\in crit(m_{\up},T_{c_{t}})$ then
$i\in C_{\up}$ by the induction hypothesis. If $i\in crit(m_{\up},T_{c_{b}})$
then $i\in crit(m_{mid},T_{c_{b}})$ and hence $i\in C_{mid}$ by
the induction hypothesis. Now the compatibility property of the top-bottom
branching case implies that $i\in C_{\up}$ (using that $i\in crit(m_{\up},T_{c_{b}})$).
If $i\in crit(m_{\up},(B(c_{t})\cup B(c_{b}))\setminus B)$ then $i\in C_{\up}$
by the compatibility property of the top-bottom branching case.

Finally, assume that $T_{c}=T_{c_{l}}\cup T_{c_{r}}\cup((B_{l}\cup B_{r})\setminus B)$
for two DP-cells defined as $c_{l}=(e,m_{\up},C_{\up},m_{abv},C_{abv},B_{l})$ and
$c_{r}=(e_{r},m_{\dn},C_{\dn},m_{abv},C_{abv},B_{r})$ in the left-right
branching case. If $i\in T_{c_{l}}$ then $b(i)>b(m_\up)$, so $i$
is not critical for
$m_{\up}$ and there is nothing to show. If $i\in crit(m_{\up},T_{c_{r}})$
then also $i\in B$ and the claim follows from the compatibility property
of the left-right branching case. Finally, if $i\in((B_{l}\cup B_{r})\setminus B$
then the claim also follows from the compatibility property.
\end{proof}

\begin{lem}[Weak feasibility] \label{lem:DP-weak-feasibility}
For each table entry $c=(e,m_{\dn},C_{\dn},m_{\up},C_{\up},B)$, we
have that $(T_{c}\cup B\cup C_{\dn}\cup C_{\up},M_{c}\cup\{m_{\dn},m_{\up}\})$
is weakly feasible. \end{lem} \begin{proof} For any edge $f$, define
$m_{f}:=m_{f}(c)$ as the highest bottleneck capacity m-task in
$M_{c}\cup\{m_{\dn}(c),m_{\up}(c),\bot\}\setminus\{\top\}$
using edge $f$. Let also $T_{c}^{ext}:=T_{c}\cup C_{\dn}(c)\cup C_{\up}(c)\cup B(c)$.
With this notation, we need to prove that for each edge $f$ 
\[
d(\abv_{f}(m_{f}(c),T_{c}^{ext}))+d(\crit_{f}(m_{f}(c),T_{c}^{ext}))+d(m_{f}(c))\leq u_{f}.
\]

We prove the claim by induction, following the partial order $\prec$
on the cells. If $c$ is a base case cell then $T_{c}\subseteq T(c)=\emptyset$,
$M_{c}\subseteq M(c)=\emptyset$ and hence $T_{c}^{ext}=C_{\dn}\cup C_{\up}\cup B$.
By the consistency property, $(T_{c}^{ext},\{m_{\dn},m_{\up}\})$
is weakly feasible.

For notation convenience, let us say that $e'<e''$ if edge $e'$
is to the left of edge $e''$ and $e'\ne e''$. We define analogously
$\leq$, $>$, and $\geq$.

Suppose now that $c$ is not a base case cell. By induction hypothesis,
we know that the claim holds for any cell $c'\prec c$. Assume w.l.o.g.
that $e<e_{\dn}$. Let $e_{r}$ be the first edge to the right of
$e$ (possibly $e_{r}=e_{\dn}$). 
We distinguish $3$ cases, depending on the branching that defines
the maximum value of $w(c)$.

\smallskip\noindent
\textbf{a) (Single branching)} Let $c_{s}$ be the cell achieving
the maximum. Recall that $T_{c}=T_{c_{s}}\cup(B_{s}\setminus B)$ and $M_{c}=M_{c_{s}}$.
We have $m_{f}=m_{f}(c)=m_{f}(c_{s})$ because $M_{c}=M_{c_{s}}$.
Let us assume $e_{\up}<e$, the case $e_{\up}\geq e$ being analogous.
Consider any edge $f$. We distinguish $3$ subcases depending on
the relative position of $f$:

\smallskip\noindent
\textbf{a.1) ($\mathbf{\Bf\leq \Be_{\up}}$ or $\mathbf{\Bf\geq \Be_{\dn}}$)}
Here $T(c)\cap T_{f}=M(c)\cap M_{f}=\emptyset$, hence $T_{c}^{ext}=C_{\dn}\cup C_{\up}\cup B$
and $M_{c}\cup\{m_{\dn},m_{\up}\}=\{m_{\dn},m_{\up}\}$. The claim
follows by the consistency property.

\smallskip\noindent
\textbf{a.2) ($\mathbf{\Be<\Bf<\Be_{\dn}}$)}. In this range of edges we
have $(B\setminus B_{s})\cap T_{f}=\emptyset$ by the compatibility property
of the single branching case. Hence $T_{c}^{ext}\cap T_{f}=(T_{c}\cup
C_{\dn}\cup C_{\up}\cup B)\cap T_{f}=(T_{c_{s}}\cup B\cup B_{s}\cup
C_{\dn}\cup C_{\up})\cap T_{f}=(T_{c_{s}}\cup C_{\dn}\cup
C_{\up}\cup B_{s})\cap T_{f}=T_{c_{s}}^{ext}\cap T_{f}$.
As a consequence, $\abv_{f}(m_{f},T_{c}^{ext})=\abv_{f}(m_{f},T_{c_{s}}^{ext})$
and $\crit_{f}(m_{f},T_{c}^{ext})=\crit_{f}(m_{f},T_{c_{s}}^{ext})$.
The claim follows by induction hypothesis on $c_{s}$.

\smallskip\noindent
\textbf{a.3) ($\mathbf{\Be_{\up}<\Bf\leq \Be}$)} In this case $b(m_{f})\geq b(m_{\up})$.
Since any task $i\in B$ has $b(i)\leq b(m_{\up})$, we
have $\abv_{f}(m_{f},T_{c}^{ext})=\abv_{f}(m_{f},T_{c_{s}}\cup B_{s}\cup
B\cup C_{\dn}\cup C_{\up})=\abv_{f}(m_{f},T_{c_{s}}\cup B_{s}\cup
C_{\dn}\cup C_{\up})=\abv_{f}(m_{f},T_{c_{s}}^{ext})$.
Also, any task $i\in B$ that is critical for $m_{f}$ must be contained
in $C_{\up}$ by the compatibility property of the single branching,
hence $\crit_{f}(m_{f},B)\subseteq C_{\up}\cap T_{f}$. Therefore
$\crit_{f}(m_{f},T_{c}^{ext})=\crit_{f}(m_{f},T_{c_{s}}\cup B_{s}\cup
B\cup C_{\dn}\cup C_{\up})=\crit_{f}(m_{f},T_{c_{s}}\cup B_{s}\cup
C_{\dn}\cup C_{\up})=\crit_{f}(m_{f},T_{c_{s}}^{ext})$.
The claim then follows by induction hypothesis on $c_{s}$.

\smallskip\noindent
\textbf{b) (Top-bottom branching)} Let $c_{b}$ and $c_{t}$ be the
cells achieving the maximum. Recall that $M_{c}=M_{c_{b}}\cup
M_{c_{t}}\cup\{m_{mid}\}$
and $T_{c}=T_{c_{b}}\cup T_{c_{t}}\cup((B_{b}\cup B_{t})\setminus B)$. Let
$e_{mid}:=e(m_{mid})$. Note that $e<e_{mid}$ and $e_{\up}<e_{mid}$.
Let us assume $e_{\up}<e$, the case $e_{\up}\geq e$ being analogous.
Consider any edge $f$. We distinguish $4$ subcases:

\smallskip\noindent 
\textbf{b.1) ($\mathbf{\Bf\leq \Be_{\up}}$ or $\mathbf{\Bf\geq \Be_{\dn}}$)}
Here $T(c)\cap T_{f}=M(c)\cap M_{f}=\emptyset$. 
The claim follows by the same argument as in case (a.1).

\smallskip\noindent
\textbf{b.2) ($\mathbf{\Be_{mid}\leq \Bf<\Be_{\dn}}$)} Note that $M(c_{t})\cap
M_{f}=T(c_{t})\cap T_{f}=\emptyset$.
We have $m_{f}=m_{f}(c)=m_{f}(c_{b})$. Observe also that $((B\cup
B_{t})\setminus B_{b})\cap T_{f}=\emptyset$
and $C_{\up}\cap T_{f}\subseteq C_{mid}\cap T_{f}$ by the compatibility
property of the top-bottom branching case. Altogether $T_{c}^{ext}\cap
T_{f}=(T_{c_{b}}\cup T_{c_{t}}\cup B_{b}\cup B_{t}\cup B\cup C_{\up}\cup
C_{\dn})\cap T_{f}=(T_{c_{b}}\cup B_{b}\cup C_{\up}\cup C_{\dn})\cap
T_{f}\subseteq(T_{c_{b}}\cup B_{b}\cup C_{mid}\cup C_{\dn})\cap
T_{f}=T_{c_{b}}^{ext}\cap T_{f}$.
Then $\abv_{f}(m_{f},T_{c}^{ext})\subseteq\abv_{f}(m_{f},T_{c_{b}}^{ext})$
and $\crit_{f}(m_{f},T_{c}^{ext})\subseteq\crit_{f}(m_{f},T_{c_{b}}^{ext})$.
The claim follows by induction hypothesis on $c_{b}$.

\smallskip\noindent
\textbf{b.3) ($\mathbf{\Be_{r}\leq \Bf<\Be_{mid}}$)} We have
$m_{f}=m_{f}(c)=m_{f}(c_{t})$
and observe that $b(m_{f})\geq b(m_{mid})$. By the compatibility
property of the top-bottom branching, if $i\in B\cap T_{f}$ and $b(i)>b(m_{mid})$,
then $i\in B_{t}$. Note also that for any $i\in(C_{mid}\cup C_{\dn}\cup
T_{c_{b}}\cup B_{b})\cap T_{f}$
we have $b(i)\leq b(m_{mid})$. Altogether
$\abv_{f}(m_{f},T_{c}^{ext})=\abv_{f}(m_{f},T_{c_{b}}\cup T_{c_{t}}\cup
B_{b}\cup B_{t}\cup B\cup C_{\dn}\cup
C_{\up})=\abv_{f}(m_{f},T_{c_{t}}\cup B_{t}\cup
C_{\up})=\abv_{f}(m_{f},T_{c_{t}}\cup B_{t}\cup C_{mid}\cup
C_{\up})=\abv_{f}(m_{f},T_{c_{t}}^{ext})$.
By the compatibility property of the top-bottom branching, if
$i\in(B_{b}\cup B\cup C_{\dn})\cap T_{f}$
is critical for $m_{f}$ (hence for $m_{mid}$), then $i\in C_{mid}$.
By Lemma~\ref{lem:DP-inclusion}, if $i\in T_{c_{b}}\cap T_{f}$ is
critical for $m_{f}$ (hence for $m_{mid}$), then $i\in C_{mid}$.
Altogether $\crit_{f}(m_{f},T_{c}^{ext})=\crit_{f}(m_{f},T_{c_{b}}\cup
T_{c_{t}}\cup B_{b}\cup B_{t}\cup B\cup C_{\dn}\cup
C_{\up})\subseteq\crit_{f}(m_{f},T_{c_{t}}\cup B_{t}\cup C_{\up}\cup
C_{mid})=\crit_{f}(m_{f},T_{c_{t}}^{ext})$.
The claim follows by induction hypothesis on $c_{t}$.

\smallskip\noindent
\textbf{b.4) ($\mathbf{\Be_{\up}<\Bf<\Be_{r}}$)} We have $m_{f}=m_{f}(c)=m_{f}(c_{t})$
and observe that $b(m_{f})\geq b(m_{\up})$. Note that $T_{c_{b}}\cap
T_{f}=\emptyset$.
Also, for any $i\in(B\cup B_{b}\cup C_{\dn}\cup C_{mid})\cap T_{f}$
we have that $b(i)\leq b(m_{\up})$. Then
$\abv_{f}(m_{f},T_{c}^{ext})=\abv_{f}(m_{f},T_{c_{b}}\cup T_{c_{t}}\cup
B_{b}\cup B_{t}\cup B\cup C_{\dn}\cup
C_{\up})=\abv_{f}(m_{f},T_{c_{t}}\cup B_{t}\cup
C_{\up})=\abv_{f}(m_{f},T_{c_{t}}^{ext})$.
By the compatibility property of the top-bottom branching case, if
$i\in(B_{b}\cup B\cup C_{\dn}\cup C_{mid})\cap T_{f}$ is critical
for $m_{f}$, then $i\in C_{\up}$. Thus
$\crit_{f}(m_{f},T_{c}^{ext})=\crit_{f}(m_{f},T_{c_{b}}\cup
T_{c_{t}}\cup B_{b}\cup B_{t}\cup B\cup C_{\dn}\cup
C_{\up})=\crit_{f}(m_{f},T_{c_{t}}\cup B_{t}\cup C_{\up}\cup
C_{mid})=\crit_{f}(m_{f},T_{c_{t}}^{ext})$.
The claim follows by induction hypothesis on $c_{t}$.

\smallskip\noindent
\textbf{c) (Left-right branching)} Let $c_{l}$ and $c_{r}$ be the
cells achieving the maximum. Recall that $M_{c}=M_{c_{l}}\cup
M_{c_{r}}\cup\{m_{abv}\}$
and $T_{c}=T_{c_{l}}\cup T_{c_{r}}\cup((B_{l}\cup B_{r})\setminus B)$. Let
$e_{abv}:=e(m_{abv})$ ($e_{abv}:=e_{r}$ if $m_{abv}=\top$). Let
us assume that $e_{abv}>e_{r}$, the case $e_{abv}\leq e_{r}$ being
analogous. Consider any edge $f$. We distinguish $4$ subcases:

\smallskip\noindent
\textbf{c.1) ($\mathbf{\Bf\leq \Be_{\up}}$ or $\mathbf{f\geq e_{\dn}}$)}
In this case $T(c)\cap T_{f}=M(c)\cap M_{f}=\emptyset$.
The claim follows by the same argument as in case
(a.1).

\smallskip\noindent
\textbf{c.2) ($\mathbf{\Be_{abv}\leq f<\Be_{\dn}}$)} In this case
$T(c_{l})\cap T_{f}=M(c_{l})\cap M_{f}=\emptyset$.
As a consequence, $m_{f}(c)=m_{f}(c_{r})$. Also, $((B\cup
B_{l})\setminus B_{r})\cap T_{f}=\emptyset$
by the compatibility property of the left-right branching. Furthermore,
$C_{\up}\cap T_{f}\subseteq B_{r}\cup C_{\dn}$ and $C_{abv}\cap
T_{f}\subseteq B_{r}\cup T_{c_{r}}\cup C_{\dn}$.
Then $T_{c}^{ext}\cap T_{f}=(T_{c_{l}}\cup T_{c_{r}}\cup B_{l}\cup
B_{r}\cup B\cup C_{\dn}\cup C_{\up})\cap T_{f}=(T_{c_{r}}\cup
B_{r}\cup C_{\dn})\cap T_{f}=(T_{c_{r}}\cup B_{r}\cup C_{\dn}\cup
C_{abv})\cap T_{f}=T_{c_{r}}^{ext}\cap T_{f}$.
As a consequence
$\abv_{f}(m_{f}(c),T_{c}^{ext})=\abv_{f}(m_{f}(c_{r}),T_{c_{r}}^{ext})$
and $\crit_{f}(m_{f}(c),T_{c}^{ext})=\crit_{f}(m_{f}(c_{r}),T_{c_{r}}^{ext})$.
The claim follows by induction hypothesis on $c_{r}$.

\smallskip\noindent
\textbf{c.3) ($\mathbf{\Be_{r}\leq \Bf<\Be_{abv}}$)} We have
$m_{f}=m_{f}(c)=m_{f}(c_{l})$
and $b(m_{f})\geq b(m_{abv})$. Note that any task $i\in(T_{c_{r}}\cup
B\cup B_{r}\cup C_{\dn}\cup C_{abv})\cap T_{f}$
has $b(i)\leq b(m_{abv})$. Hence
$\abv_{f}(m_{f},T_{c}^{ext})=\abv_{f}(m_{f},T_{c_{l}}\cup T_{c_{r}}\cup
B_{l}\cup B_{r}\cup B\cup C_{\dn}\cup
C_{\up})=\abv_{f}(m_{f},T_{c_{l}}\cup B_{l}\cup C_{\up}\cup
C_{abv})=\abv_{f}(m_{f},T_{c_{l}}^{ext})$.
Furthermore, if a task $i\in(T_{c_{r}}\cup B\cup B_{r}\cup C_{\dn})\cap T_{f}$
is critical for $m_{f}$, then $i\in C_{abv}$ by the compatibility
property of the left-right branching. Consequently
$\crit_{f}(m_{f},T_{c}^{ext})=\crit_{f}(m_{f},T_{c_{l}}\cup
T_{c_{r}}\cup B_{l}\cup B_{r}\cup B\cup C_{\dn}\cup
C_{\up})\subseteq\crit_{f}(m_{f},T_{c_{l}}\cup C_{abv}\cup
C_{\up}\cup B_{l})=\crit_{f}(m_{f},T_{c_{l}}^{ext})$.
The claim follows by induction hypothesis on $c_{l}$.

\smallskip\noindent
\textbf{c.4) ($\mathbf{\Be_{\up}<\Bf<\Be_{r}}$)} In this case
$m_{f}=m_{f}(c)=m_{f}(c_{l})$
and $b(m_{f})\geq b(m_{\up})$. By the compatibility property of the
left-right branching, $((B_{r}\cup B)\setminus B_{l})\cap T_{f}=B\cap T_{f}$.
Observe that any task $i\in(B\cup C_{\dn})\cap T_{f}$ has 
$b(i)\leq b(m_{\up})$. Also, $T_{c_{r}}\cap T_{f}\subseteq T(c_{r})\cup
T_{f}=\emptyset$.
Then $\abv_{f}(m_{f},T_{c}^{ext})=\abv_{f}(m_{f},T_{c_{l}}\cup
T_{c_{r}}\cup B_{l}\cup B_{r}\cup B\cup C_{\dn}\cup
C_{\up})=\abv_{f}(m_{f},T_{c_{l}}\cup B_{l}\cup
C_{\up})\subseteq\abv_{f}(m_{f},T_{c_{l}}\cup B_{l}\cup
C_{\up}\cup C_{abv})=\abv_{f}(m_{f},T_{c_{l}}^{ext})$.
By the compatibility property of the left-right branching, any
$i\in(B\cup C_{\dn})\cap T_{f}$
that is critical for $m_{f}$ must belong to $C_{\up}$. Thus
$\crit_{f}(m_{f},T_{c}^{ext})=\crit_{f}(m_{f},T_{c_{l}}\cup
T_{c_{r}}\cup B_{l}\cup B_{r}\cup B\cup C_{\dn}\cup
C_{\up})=\crit_{f}(m_{f},T_{c_{l}}\cup B_{l}\cup
C_{\up})\subseteq\crit_{f}(m_{f},T_{c_{l}}\cup B_{l}\cup
C_{\up}\cup C_{abv})=\crit_{f}(m_{f},T_{c_{l}}^{ext})$.
The claim follows by induction hypothesis on $c_{l}$. \end{proof}

Now the proof of Lemma~\ref{lem:DP} follows from Lemmas~\ref{lem:DP-running-time}, ~\ref{lem:DP-optimality},
\ref{lem:DP-k-thinness}, and~\ref{lem:DP-weak-feasibility}, and the
fact that the cell $c^{*}:=(e^{*},\bot,\emptyset,\top,\emptyset,\emptyset)$ corresponds
to the optimal weakly-feasible $k$-thin maze-pair.

\bibliographystyle{plain}
\bibliography{citations}

\begin{thebibliography}{10}

\bibitem{AzarRegev2006}
Y.~Azar and O.~Regev.
\newblock Combinatorial algorithms for the unsplittable flow problem.
\newblock {\em Algorithmica}, 44:49--66, 2006.

\bibitem{BCES2006}
N.~Bansal, A.~Chakrabarti, A.~Epstein, and B.~Schieber.
\newblock A quasi-{PTAS} for unsplittable flow on line graphs.
\newblock In {\em STOC}, pages 721--729. ACM, 2006.

\bibitem{SODA-unsplit-flow}
N.~Bansal, Z.~Friggstad, R.~Khandekar, and R.~Salavatipour.
\newblock A logarithmic approximation for unsplittable flow on line graphs.
\newblock In {\em SODA}, pages 702--709, 2009.

\bibitem{BSW11}
P.~Bonsma, J.~Schulz, and A.~Wiese.
\newblock A constant factor approximation algorithm for unsplittable flow on
  paths.
\newblock In {\em FOCS}, pages 47--56, 2011.

\bibitem{CCKR2002}
G.~Calinescu, A.~Chakrabarti, H.~Karloff, and Y.~Rabani.
\newblock Improved approximation algorithms for resource allocation.
\newblock In {\em IPCO}, pages 401--414, 2002.

\bibitem{CCGK2007}
A.~Chakrabarti, C.~Chekuri, A.~Gupta, and A.~Kumar.
\newblock Approximation algorithms for the unsplittable flow problem.
\newblock {\em Algorithmica}, pages 53--78, 2007.

\bibitem{CEKunp}
C.~Chekuri, A.~Ene, and N.~Korula.
\newblock Unsplittable flow in paths and trees and column-restricted packing
  integer programs.
\newblock Unpublished. Available at
  \url{http://web.engr.illinois.edu/~ene1/papers/ufp-full.pdf}.

\bibitem{CEKApprox2009}
C.~Chekuri, A.~Ene, and N.~Korula.
\newblock Unsplittable flow in paths and trees and column-restricted packing
  integer programs.
\newblock In {\em APPROX-RANDOM}, pages 42--55, 2009.

\bibitem{CMS07}
C.~Chekuri, M.~Mydlarz, and F.~Shepherd.
\newblock Multicommodity demand flow in a tree and packing integer programs.
\newblock {\em ACM Transactions on Algorithms}, 3, 2007.

\bibitem{CWMX-ESA2010}
M.~Chrobak, G.~Woeginger, K.~Makino, and H.~Xu.
\newblock Caching is hard, even in the fault model.
\newblock In {\em ESA}, pages 195--206, 2010.

\bibitem{DPS2010}
A.~Darmann, U.~Pferschy, and J.~Schauer.
\newblock Resource allocation with time intervals.
\newblock {\em Theoretical Computer Science}, 411:4217--4234, 2010.

\bibitem{GVY1997}
N.~Garg, V.~V. Vazirani, and M.~Yannakakis.
\newblock Primal-dual approximation algorithms for integral flow and multicut
  in trees.
\newblock {\em Algorithmica}, 18(1):3--20, 1997.

\bibitem{GuruswamiKhanna2003}
V.~Guruswami, S.~Khanna, R.~Rajaraman, B.~Shepherd, and M.~Yannakakis.
\newblock Near-optimal hardness results and approximation algorithms for
  edge-disjoint paths and related problems.
\newblock {\em Journal of Computer and System Sciences}, 67(3):473 -- 496,
  2003.

\bibitem{Kleinberg96}
J.~M. Kleinberg.
\newblock {\em Approximation Algorithms for Disjoint Paths Problems}.
\newblock PhD thesis, MIT, 1996.

\bibitem{Schrijver03}
A.~Schrijver.
\newblock {\em Combinatorial Optimization: Polyhedra and Efficiency}.
\newblock Springer, Berlin, 2003.

\end{thebibliography}

\newpage

\appendix
\section{PTAS for $\delta$-small tasks}
\label{apx:PTAS-small-tasks}

The techniques in \cite{BSW11} immediately imply a $(1+\eps)$-approximation
algorithm for $\delta$-small tasks, assuming that $\delta$ is sufficiently small,
depending on $\eps$.

Fix an integer $\ell \in \mathbb{N}$. For each $k\in \mathbb{N}$ we define a set
$F^{k,\ell}:=\left\{ i\in T |  2^{k}\le b(i)<2^{k+\ell}\right\}$. Let $OPT(F^{k,\ell})$ be the optimum solution when considering only tasks in $F^{k,\ell}$. For each of these sets,
we want to compute a set that approximates the optimal solution well while leaving
a bit of the capacity of the edges unused.

\begin{defn}[\cite{BSW11}] 
Consider a set $F^{k,\ell}$ and let \mbox{$\alpha,\beta>0$}. A set $F\subseteq F^{k,\ell}$ is
called $(\alpha,\beta)$\emph{-approximative} if 
\begin{itemize}
\item $w(F)\ge\frac{1}{\alpha}\cdot w(OPT(F^{k,\ell}))$, and 
\item $\sum_{i\in F\cap T_{e}}d_{i}\le u_{e}-\beta\cdot2^{k}$ for each
edge $e$ such that $T_{e}\cap F^{k,\ell}\ne\emptyset$.  (Hence it is a feasible solution.)
\end{itemize}
An algorithm that computes \emph{$(\alpha,\beta)$-approximative}
sets in polynomial time is called an $(\alpha,\beta)$\emph{-approximation
algorithm}. 
\end{defn}
\smallskip 

The reader may think of $\beta$ being a constant in the same magnitude as $\eps$. 

\begin{lem*}[\cite{BSW11}]
\label{lem:F^k-LP}
For every combination of constants \mbox{$\eps>0$}, ${0<\beta<1}$, and $\ell\in\mathbb{N}$, there exists a $\delta>0$ such that if all tasks are $\delta$-small, then for each set $F^{k,\ell}$ there is a polynomial time $(\frac{1+\eps}{1-\beta},\beta)$-approximation algorithm.
\end{lem*}

The next lemma shows that while losing only a factor $\frac{\ell+q}{\ell}$ in the approximation ratio,
$(\alpha,\beta)$-approximation algorithms yield approximation algorithms for the whole problem. By choosing
$\ell$ large enough depending on $q$, this loses only a factor of $1+\eps$.

\begin{lem*}[\cite{BSW11}]
\label{lem:framework} Let $\ell \in \mathbb{N}$ and $q \in \mathbb{N}$ be constants and let 
$\beta=2^{1-q}$.
Assume we are given an \emph{$(\alpha,\beta)$}-approximation algorithm
for each set $F^{k,\ell}$, with running time $O(p(n))$ for a polynomial
$p$. Then there is a $\left(\frac{\ell+q}{\ell}\cdot\alpha\right)$-approximation
algorithm with running time $O(m\cdot p(n))$ for the set of \emph{all}
tasks.
\end{lem*}

By defining the constants appropriately we obtain Lemma~\ref{lem:ptasSmall}.
For a given $\epsilon>0$ we choose $q \in \mathbb{N}$ such that $\beta:=2^{1-q}\le \epsilon$.
Then we define $\ell\in \mathbb{N}$ such that $\frac{\ell+q}{\ell}\le 1+\epsilon$.
Lemma~\ref{lem:F^k-LP} yields a value $\delta>0$. Using Lemma~\ref{lem:framework} we get
a $\frac{1+\epsilon}{1-\beta}\cdot \frac{\ell+q}{\ell}\le 1+O(\epsilon)$-approximation algorithm
for $\delta$-small tasks. As the above reasoning holds for any $\epsilon>0$, the claim  
of Lemma~\ref{lem:ptasSmall} follows.

\end{document}